\documentclass[english,aps,prl,superscriptaddress,floatfix,notitlepage,reprint]{revtex4-1}
\usepackage[T1]{fontenc}
\usepackage[latin9]{inputenc}
\usepackage{color}
\usepackage{mathrsfs}
\usepackage{bm}
\usepackage{amsmath}
\usepackage{amsthm}
\usepackage{amssymb}
\usepackage{graphicx}

\makeatletter
\theoremstyle{plain}
\newtheorem{thm}{\protect\theoremname}
\theoremstyle{plain}
\newtheorem{lem}[thm]{\protect\lemmaname}
\theoremstyle{plain}
\newtheorem{cor}[thm]{\protect\corollaryname}

\usepackage{babel}
\usepackage{amsthm}
\usepackage{braket}
\usepackage{txfonts}
\usepackage{graphicx}
\usepackage[usenames,dvipsnames]{xcolor}
\usepackage[colorlinks=true,citecolor=Blue,linkcolor=RubineRed,urlcolor=Blue]{hyperref}
\usepackage{bbm}
\renewcommand{\textcolor}[2]{#2}
\date{\today}

\makeatother

\usepackage{babel}
\providecommand{\corollaryname}{Corollary}
\providecommand{\lemmaname}{Lemma}
\providecommand{\theoremname}{Theorem}

\begin{document}
\title{Super-Heisenberg scaling in Hamiltonian parameter estimation in the
long-range Kitaev chain}
\author{Jing Yang}
\email{jyang75@ur.rochester.edu}

\affiliation{Department of Physics and Astronomy, University of Rochester, Rochester,
New York 14627, USA}
\address{Department of Physics and Materials Science, University of Luxembourg,
L-1511 Luxembourg, Luxembourg}
\author{Shengshi Pang}
\email{pangshsh@mail.sysu.edu.cn}

\affiliation{School of Physics, Sun Yat-Sen University, Guangzhou, Guangdong 510275,
China}

\author{Adolfo del Campo}
\email{adolfo.delcampo@uni.lu}

\address{Department of Physics and Materials Science, University of Luxembourg,
L-1511 Luxembourg, Luxembourg}
\address{Donostia International Physics Center, E-20018 San Sebasti\'an, Spain}

\author{Andrew N. Jordan}
\email{jordan@pas.rochester.edu}

\affiliation{Institute for Quantum Studies, Chapman University, 1 University Drive,
Orange, CA 92866, USA}
\affiliation{Department of Physics and Astronomy, University of Rochester, Rochester,
New York 14627, USA}

\begin{abstract}
In quantum metrology, nonlinear many-body interactions can enhance
the precision of Hamiltonian parameter estimation to surpass the Heisenberg
scaling. Here, we consider  the estimation of the interaction strength
in linear systems with long-range interactions and using the Kitaev
chains as a case study, we establish a transition from the Heisenberg
to super-Heisenberg scaling in the quantum Fisher information by varying
the interaction range. We further show that quantum control can improve
the prefactor of the quantum Fisher information. Our results explore
the advantage of optimal quantum control and long-range interactions
in many-body quantum metrology. 
\end{abstract}
\maketitle

\section{Introduction}

\textcolor{blue}{Quantum metrology is a paradigmatic example of an
emergent technology in which quantum resources can provide an advantage
with no classical counterpart~\citep{helstrom1968theminimum,helstrom1976quantum2,holevo2011probabilistic,braunstein_statistical_1994}.
The general scheme for estimating the parameters in a Hamiltonian
system is depicted in Fig.~\ref{fig:protocol}. Quantum mechanics
provides two key ingredients improving the precision: (a) the coherence
in state $\rho_{\theta}$ of the probe, which is controlled by the
probe time $T$, and (b) entanglement, when $N$ probes are allowed
in a single round, which can be introduced in the initial state or
generated via many-body interactions during the sensing process. According
to the quantum Cramer-Rao bound, the uncertainty $\delta\theta$ in
the estimation of the parameter $\theta$ in Fig.~\ref{fig:protocol}
is governed by the quantum Fisher information $I(\theta)$ (QFI) as
$\delta\theta\geq1/\sqrt{\nu I(\theta)}$, where $\nu$ is the number
of repetitions of the process. The QFI plays a fundamental role in
the geometry of the space of quantum states and has manifold applications,
which include witnessing a quantum phase transition~\citep{quan2006decayof,zanardi2007informationtheoretic,camposvenuti2007quantum,gu2010fidelity},
critical sensing~\citep{rams2018atthe,chu2021dynamic,garbe2020critical,mishra2020driving}
and detecting multi-partite entanglement~\citep{toth2012multipartite,hyllus2012fisherinformation,pezz`e2017multipartite}.}
\begin{figure}[t]
\begin{centering}
\includegraphics[scale=0.37]{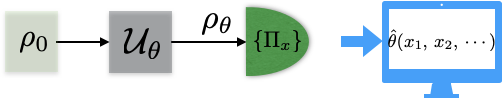}
\par\end{centering}
\caption{\label{fig:protocol}A typical round in quantum metrology consisting
of four steps (i) Preparing an initial quantum state $\rho_{0}$ (ii)
Evolving the initial state $\rho_{0}$ under a parameter $\theta$-dependent
unitary quantum channels $\mathcal{U}_{\theta}$ to obtain a parameter-dependent
state $\rho_{\theta}$ (iii) Performing a quantum measurement described
by the positive operator-valued measure operators $\{\Pi_{x}\}$ on
the state $\rho_{\theta}$ to get data $x_{n}$ (iv) Steps (i)-(iii)
can be run multiple times in parallel or a sequential scheme~\citep{giovannetti2006quantum,boixo2007generalized},
which generates a large number of measurement data $\{x_{1},\,x_{2},\cdots\}$.
Processing the data with the maximum likelihood estimator $\hat{\theta}(x_{1},\,x_{2},\,\cdots)$
saturates the classical-Cram\'er-Rao bound.}
\end{figure}

Recently, optimal control has been shown to offer a new arena for
enhancing quantum parameter estimation~\citep{yuan2015optimal,pang2017optimal,yang2017quantum}.
The interplay between quantum control theory and quantum many-body
systems is yet to be undertaken and it is crucial to understand quantum
parameter estimation of coupling constants in realistic systems with
long-range interactions.\textcolor{blue}{{} In Hamiltonian parameter
estimation of noninteracting spin systems, the maximum possible QFI
scales linearly with the number of probes, for an uncorrelated initial
state~\citep{giovannetti2006quantum}. The scaling becomes quadratic
if the probes are initially prepared in the GHZ state with maximum
entanglement, known as the Heisenberg scaling (HS)~\citep{giovannetti2006quantum}.
This naturally motivates the idea of surpassing the HS, reaching the
so-called super-HS scaling, by introducing nonlinear interactions
in the sensing Hamiltonian ~\citep{beau2017nonlinear,boixo2007generalized,roy2008exponentially}.
These works have led to the intuitive belief that surpassing the HS
requires nonlinear interactions. }

\textcolor{blue}{In this paper, we explore Hamiltonain parameter estimation
in linear systems with long-range interactions, using a case study
of the generalization of the Long-Range Kitaev (LRK) chain~\citep{vodola2014kitaevchains,vodola2015longrange,viyuela2016topological}
to allow for general decay laws of the long-range interactions. By
focusing on the estimation of the long-range superconducting strength,
we establish that super-HS can be achieved in the case of slowly decaying
}\textit{\textcolor{blue}{linear}}\textcolor{blue}{{} long-range
interactions. Indeed, we observe a transition from HS to super-HS
for a specific value of the exponent governing the decay law of the
interactions. In all cases, quantum control may improve the prefactor
of the scaling of the QFI as a function of the the number of lattice
sites.}

\section{Hamiltonian estimation of\textit{ }the LRK model.}

We consider parameter estimation with a general time-dependent Hamiltonian
$H_{\theta}(t)$, where $\theta$ is the estimation parameter and
the parametric dependence is general, i.e., not necessarily multiplicative.
The effective generator for the parameter estimation is defined as
$\ket{\psi_{\theta}}=\text{e}^{-\text{i}G_{\theta}}\ket{\psi_{0}}$~\citep{boixo2007generalized},
where $\ket{\psi_{0}}$ is the initial state and $\ket{\psi_{\theta}}$
is the effective parameter-dependent state which gives the same QFI
as the true physical state~\footnote{Note that unless the Hamiltonian is multiplicative, $\ket{\psi_{\theta}}$
in general may be not the same as the true physical state $\mathcal{U}_{\theta}(T)\ket{\psi_{0}}$.
However, $\ket{\psi_{\theta}}$ is able to give the same QFI as the
true physical state, see Ref.~\citep{pang2017optimal} for further
justifications.}. For a general Hamiltonian it is given by~\citep{pang2017optimal,pang2014quantum,boixo2007generalized}
\[
G_{\theta}=\int_{0}^{T}\mathcal{U}_{\theta}^{\dagger}(\tau)\partial_{\theta}H_{\theta}(\tau)\mathcal{U}_{\theta}(\tau)d\tau,
\]
 where $\mathcal{U}_{\theta}(\tau)$ is the evolution operator. Once
the generator is obtained, the quantum Fisher information is given
by $I(\theta)=4\text{Var}[G_{\theta}]|_{\ket{\psi_{0}}}.$ Maximization
over all the possible initial states gives 
\begin{equation}
I(\theta)=[\vartheta_{\max}(T)-\vartheta_{\min}(T)]^{2},\label{eq:I}
\end{equation}
where $\ket{\vartheta_{\min}(T)}$ and $\ket{\vartheta_{\max}(T)}$
as the eigenvectors that corresponding to the minimum and maximum
eigenvalues of $G_{\theta}$ and the corresponding initial state is
prepared in an equal superposition between $\ket{\vartheta_{\max}(T)}$
and $\ket{\vartheta_{\min}(T)}$. When coherent optimal controls are
possible, one can further optimize the QFI over the unitary dynamics
appear in the generator $G_{\theta}$. We denote the eigenstates of
$\partial_{\theta}H_{\theta}(t)$ at the instant time $t$ as $\ket{\chi_{n}(t)}$.
It turns out the optimal unitary dynamics is the one which steers
the state always towards $\ket{\chi_{n}(t)}$, if one starts with
$\ket{\chi_{0}(0)}$~\citep{pang2017optimal}. With this intuition,
it is easily found that the total Hamiltonian after including the
control is 
\begin{equation}
H_{\text{tot}}(t)=\text{i}\hbar\partial_{t}\mathcal{U}_{\text{c}\theta}(t)\mathcal{U}_{\text{c}\theta}^{-1}(t),
\end{equation}
where 
\begin{equation}
\mathcal{U}_{\text{c}\theta}(t)=\sum_{n}\ket{\chi_{n}(t)}\bra{\chi_{n}(0)}
\end{equation}
 is a unitary operator formed by the eigenvectors of $\partial_{\theta}H_{\theta}(t)$.
Therefore, the optimal control Hamiltonian is~\citep{cabedo-olaya2020shortcuttoadiabaticitylike,delcampo2013shortcuts}
\begin{equation}
H_{\text{c}}(t)=H_{\text{tot}}(t)-H_{\theta}(t).
\end{equation}
When optimal control is applied, the generator is $G_{\theta}=\sum_{n}\ket{\chi_{n}(0)}\bra{\chi_{n}(0)}\int_{0}^{T}\chi_{n}(\tau)d\tau$.
Thus the upper bound of the QFI after optimization over the initial
states and unitary dynamics is 
\begin{equation}
I_{0}(\theta)=\left(\int_{0}^{T}[\chi_{\max}(\tau)-\chi_{\min}(\tau)]d\tau\right)^{2},\label{eq:I0}
\end{equation}
where $\chi_{\max}(t)$ and $\chi_{\min}(t)$ are the maximum and
minimum eigenvalues of $\partial_{\theta}H_{\theta}(t)$. The optimal
initial state is the equal superposition between $\ket{\chi_{\min}(0)}$
and $\ket{\chi_{\min}(0)}$. For time-independent Hamiltonians, $I_{0}(\theta)$
is simply proportional to the square of the difference of the maximum
and minimum eigenvalues of $\partial_{\theta}H_{\theta}$, with the
prefactor $4T^{2}$.

Now we consider $H_{\theta}$ be the LRK Hamiltonian~\citep{vodola2014kitaevchains}
\begin{align}
H_{\theta} & =-\frac{J}{2}\sum_{j=1}^{N}(a_{j}^{\dagger}a_{j+1}+a_{j+1}^{\dagger}a_{j})-\mu\sum_{j=1}^{N}(a_{j}^{\dagger}a_{j}-\frac{1}{2})\nonumber \\
 & +\frac{\Delta}{2}\sum_{j=1}^{N-1}\sum_{l=1}^{N-j}\kappa_{\alpha,\,l}(a_{j}a_{j+l}+a_{j+l}^{\dagger}a_{j}^{\dagger}),\label{eq:LRK}
\end{align}
where consider a unit lattice spacing, $J$ represents the tunneling
rate between nearest neighbors, $\mu$ is the chemical potential,
$\Delta$ represents the strength of the $p$-wave pairing, $N$ is
the number of Fermionic lattice sites, and $\kappa_{l,\,\alpha}$
satisfies the symmetry property: $\kappa{}_{l,\,\alpha}=\kappa_{N-l,\,\alpha}$
for $1\le l\le N/2$. Here, $\alpha\ge0$ characterizes the decay
property of the long-range interaction. Without loss of generality,
we choose the normalization condition $\kappa_{1,\,\alpha}=1$. The
Kitaev chain has recently attracted broad attention as it supports
noise-resilient Majorana zero modes at its two ends for open boundary
conditions~\citep{kitaev2001unpaired,alicea2012newdirections}. Recent
works~\citep{vodola2014kitaevchains,vodola2015longrange,viyuela2016topological}
have generalized the original model to the LRK chain, which contains
long-range superconducting $p$-wave pairing, i.e., the last term
on the r.h.s of Eq.~\eqref{eq:LRK}. \textcolor{blue}{Note that by
contrast to the power-law decay law in Ref.~\citep{vodola2014kitaevchains},
we consider a general decay law $\kappa_{l,\,\alpha}$ which only
requires the existence of a finite non-negative integer $Q\ge0$ such
that} (i) $\kappa_{x,\,\alpha}$ and its derivatives with respect
to $x$ up to $2Q$ order are bounded on $[1,\,\infty]$ and (ii)
$\int_{1}^{\infty}\kappa_{x,\,\alpha}^{(2Q+1)}dx$ is finite, where
the superscript $(q)$ denotes the $q$-th derivative with respect
to $x$. Specifically, we consider power-law interactions with $\kappa_{l,\,\alpha}=l^{-\alpha}$
as in the original proposal of the LRK model~\citep{vodola2014kitaevchains},
and $\kappa_{l,\,\alpha}=(1+\ln l)^{-\alpha}$ $(\alpha\ge0)$~\footnote{Note that they represent $\kappa_{l,\,\alpha}=(\lambda l)^{-\alpha}$
with $\lambda>0$ and $\alpha\ge0$ and $\kappa_{l,\,\alpha}=[\ln(\lambda l)]^{-\alpha}$
with $\lambda>1$ and $\alpha\ge0$. Without loss of generality, we
shall set $\lambda=1$ in the former case and $\lambda=e$ in the
latter case.}. Assuming the anti-periodic boundary condition $a_{j}=-a_{j+N}$,
the LRK Hamiltonian~(\ref{eq:LRK}) can be diagonalized via the Bogoliubov
transformation~\citep{pezz`e2017multipartite}, yielding 
\begin{equation}
H_{\theta}=\sum_{k}\epsilon_{\theta}(k)\eta_{\theta}^{\dagger}(k)\eta_{\theta}(k),
\end{equation}
 where $k=\frac{1}{2}\frac{2\pi}{N},\,\frac{3}{2}\frac{2\pi}{N}\cdots\frac{2\pi}{N}(N-\frac{1}{2})$,
\begin{equation}
\epsilon_{\theta}(k)\equiv\sqrt{[\Delta f_{\alpha}(k)/2]^{2}+(J\cos k+\mu)^{2}},\label{eq:eps-theta}
\end{equation}
and 
\begin{equation}
f_{\alpha}(k)\equiv2\sum_{l=1}^{N/2-1}\kappa_{l,\,\alpha}\sin(kl)+\kappa_{N/2,\,\alpha}.
\end{equation}
The factor of $2$ in front of $f_{\alpha}(k)$ accounts for the symmetry
property of $\kappa_{l,\,\alpha}$. The generator $G_{\theta}$ can
be diagonalized via the Bogoliubov transformation as (see Appendix~\ref{sec:The-diagonaization}
for details)
\begin{equation}
G_{\theta}=\sum_{k}\mathscr{E}_{\theta}(k)\psi_{\theta}^{\dagger}(k)\psi_{\theta}(k),\label{eq:Generator-diagonal}
\end{equation}
where the Fermionic operators $\psi_{\theta}^{\dagger}(k)$ and $\psi_{\theta}(k)$
are defined in Eq.~\eqref{eq:psi-k-def} and the spectrum is 
\begin{align}
\mathscr{E}_{\theta}(k) & \equiv\{T^{2}\left[\partial_{\theta}\epsilon_{\theta}(k)\right]^{2}+\frac{1}{4}\xi_{\theta}^{2}(k)\sin^{2}[2\epsilon_{\theta}(k)T],\nonumber \\
 & +\frac{1}{4}\xi_{\theta}^{2}(k)\left(1-\cos[2\epsilon_{\theta}(k)T]\right)^{2}\}^{1/2},\label{eq:calEk}
\end{align}
with
\begin{align}
\xi_{\theta}(k) & \equiv\partial_{\theta}\cos\phi_{\theta}(k)/\sin\phi_{\theta}(k),\label{eq:xi-k-def}\\
\sin\phi_{\theta}(k) & =-\Delta f_{\alpha}(k)/[2\epsilon_{\theta}(k)],\label{eq:sin-thi-def}\\
\cos\phi_{\theta}(k) & =-(J\cos k+\mu)/\epsilon_{\theta}(k).\label{eq:cos-phi-def}
\end{align}

\section{Optimal control and optimal initial state}

We next determine the optimal controls and optimal initial states
for parameter estimation, by using the spectral properties of $\partial_{\theta}H_{\theta}(t)$,
for the different choices of the Hamiltonian parameter $\theta$.
According to Eq.~\eqref{eq:LRK-p-space} in Appendix~\ref{sec:The-diagonaization},
the representation of the LRK Hamiltonian in the momentum space, it
is readily calculated that 
\begin{align}
\partial_{J}H & =-\sum_{k}a^{\dagger}(k)a(k)\cos k,\\
\partial_{\mu}H & =-\sum_{k}a^{\dagger}(k)a(k).
\end{align}
We note that $\partial_{J}H$ and $\partial_{\mu}H$ commute with
each other. Thus, according to the preceding section, the optimal
control for estimating $J$ and $\mu$ is to cancel the long-range
superconducting terms. The maximum and minimum eigenstates for $\partial_{J}H$
are 
\begin{equation}
\ket{\mathbbm{1/2}}=\prod_{k,\,\cos k\le0}a^{\dagger}(k)\ket{0},
\end{equation}
 and 
\begin{equation}
\ket{\mathbbm{-1/2}}=\prod_{k,\,\cos k>0}a^{\dagger}(k)\ket{0},
\end{equation}
respectively. We adopt this notation since in momentum space both
the maximum and minimum eigenstates are half-occupied. The optimal
initial state for estimating $J$ under optimal control is 
\begin{equation}
\ket{\psi_{0}}=\frac{1}{\sqrt{2}}(\ket{\mathbbm{1/2}}+\ket{\mathbbm{-1/2}}).
\end{equation}
Similarly, for $\partial_{\mu}H$, the maximum eigenstate is $\ket{0}$,
the vacuum state annihilated by $a(k)$ or $a_{j}$, and the minimum
eigenstate is the fully occupied state in the momentum space, which
we denote as 
\begin{equation}
\ket{\mathbbm{1}}=\prod_{k}a^{\dagger}(k)\ket{0}.
\end{equation}
Therefore, the optimal initial state for estimating $\mu$ under optimal
control is 
\begin{equation}
\ket{\psi_{0}}=\frac{1}{\sqrt{2}}(\ket{0}+\ket{\mathbbm{1}}).
\end{equation}

The optimal control for the estimation of $\Delta$ is to cancel all
the local interaction terms, including the tunneling and kinetic terms.
We note that the diagonalization of $\partial_{\Delta}H$ is a special
case for the diagonalization of the LRK Hamiltonian, corresponding
to $J=\mu=0$ and $\Delta=1$. With this observation one finds 
\[
\partial_{\Delta}H=\frac{1}{2}\sum_{k}\big|f_{\alpha}(k)\big|b^{\dagger}(k)b(k),
\]
where $b(k)=u_{k}a(k)+v_{k}a^{\dagger}(-k)$, $u_{k}=1/2$, $v_{k}=-\text{i}/\sqrt{2}$
if $f_{\alpha}(k)\ge0$ and $v_{k}=\text{i}/\sqrt{2}$ if $f_{\alpha}(k)<0$.
We note that $v_{k}=-v_{-k}$ because $f_{\alpha}(k)$ is an odd function
of $k$. The minimum eigenstate of $\partial_{\Delta}H$ is the ground
state annihilated by $b(k)$. According to the BCS ansatz, it is 
\begin{equation}
\ket{\text{GS}}=\prod_{k}[u_{k}-v_{k}a^{\dagger}(k)a^{\dagger}(k)]\ket{0}.
\end{equation}
The maximum eigenvalue of $\partial_{\Delta}H$ corresponds to the
fully occupied state in the picture of $b(k)$ and $b^{\dagger}(k)$,
which we denote by $\ket{\text{FO}}$ and can be written as 
\begin{equation}
\ket{\text{FO}}=\prod_{k}[u_{k}^{*}-v_{k}^{*}a(k)a(-k)]\ket{\mathbbm{1}}.\label{eq:FE}
\end{equation}
One can explicitly check that $\ket{\text{FO}}$ is normalized and
satisfies $b^{\dagger}(k)\ket{\text{FO}}=0$ for all $k$. We call
Eq.~\eqref{eq:FE} the BCS-like fully occupied states since its construction
is inspired by the BCS-ground state. Thus the optimal initial state
for estimating $\Delta$ is 
\begin{equation}
\ket{\psi_{0}}=(\ket{\text{GS}}+\ket{\text{FO}})/\sqrt{2}.
\end{equation}

\subsection{HS for estimation of $J$ and $\mu$}

The difference between the maximum and minimum eigenvalues of $\partial_{J}H$
in the many-body Hilbert space is $|\cos k|$. Thus the QFI for estimating
$J$ according to Eq.~\eqref{eq:I0} is $I_{0}(J)=(\sum_{k}|\cos k|)^{2}T^{2}$.
In the limit $N\to\infty$, 
\begin{equation}
I_{0}(J)=T^{2}[(N/2\pi)^{2}\int_{0}^{2\pi}dk|\cos k|]{}^{2}=4N^{2}T^{2}/\pi^{2},\label{eq:I0-J}
\end{equation}
where we have replaced $\sum_{k}\to N/2\pi\int dk$ taking the continuum
limit, as the integrand is not singular in the integration region.
In fact, the error introduced does not scale with $N$ according to
the analysis with the Euler-Maclaurin formula in Appendix~\ref{sec:Euler-Maclaurin-formula}.
Similarly, the difference between the maximum and minimum eigenvalues
of $\partial_{\mu}H$ in the many-body Hilbert space is $1$ and therefore
$I_{0}(\mu)=N^{2}T^{2}$. We see that the scaling of the ultimate
QFI for estimating $J$ and $\mu$ is the HS. The plot of $I_{0}(J)$
with the number of lattice sites is shown in Fig.~\ref{fig:The-plots-of-IJ}.
We will show shortly that even in the case of imperfect control or
without control, such scaling is not altered.

\begin{figure}
\begin{centering}
\includegraphics[width=0.7\linewidth]{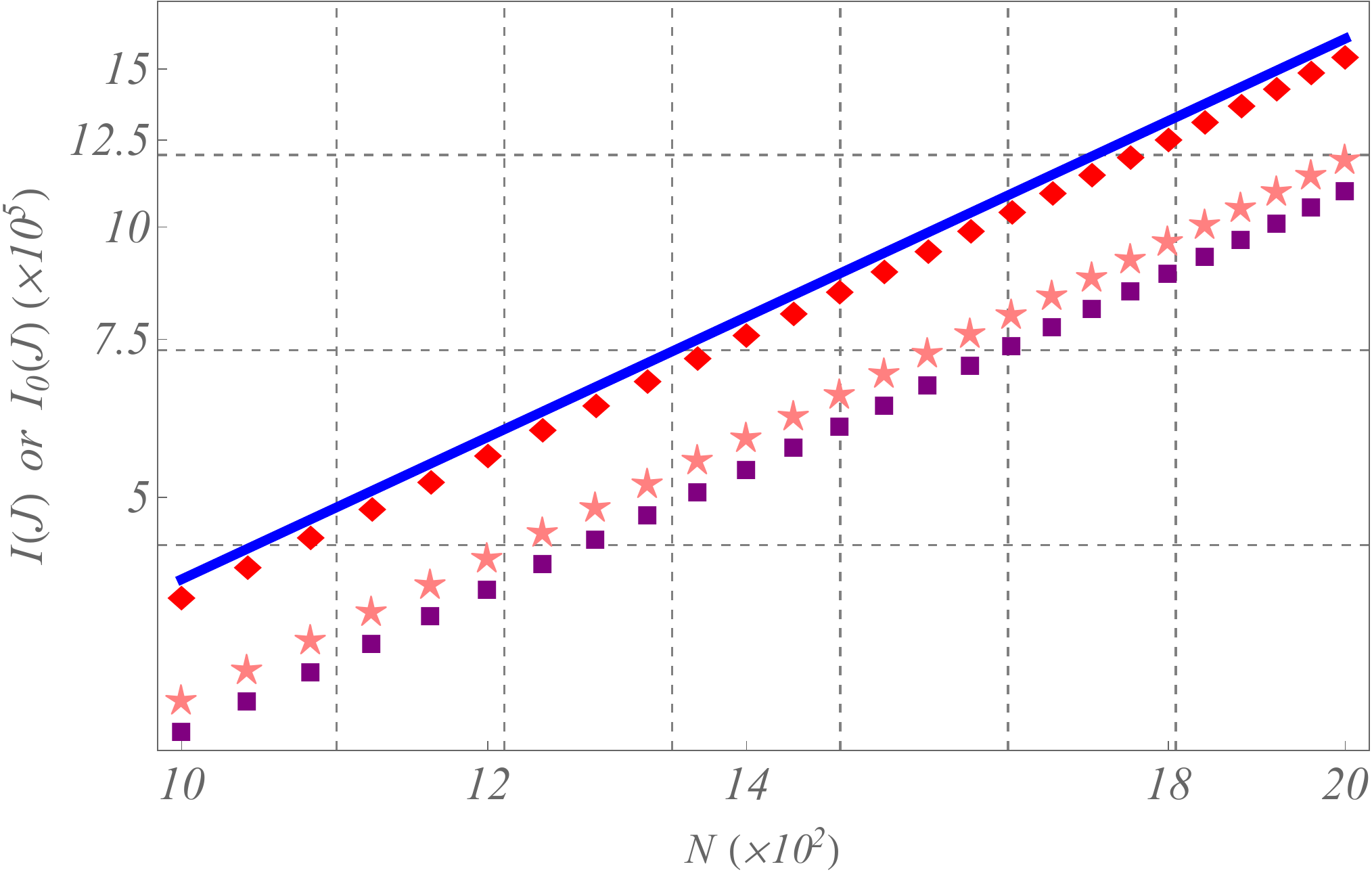} 
\par\end{centering}
\caption{\label{fig:The-plots-of-IJ} Quantum Fisher information $I(J)$ for
$J$ estimation in the case with no control or imperfect control,
as well as for optimal control, as a function of $N$ for $\kappa_{l,\,\alpha}=l^{-\alpha}$.
The blue solid line is the maximum possible QFI plotted according
to the analytical expression Eq.~\eqref{eq:I0-J} when optimal controls
are applied, where the long-range interaction terms are cancelled
(with probe time $T=1$). All the discrete points are numerically
calculated from Eq.~\eqref{eq:I-theta}. The values of the parameters
are (i) red circle dots: $J=\mu=10\Delta$ and $\alpha=0$, (ii) pink
stars are $J=\mu=\Delta$ and $\alpha=0.5$, and (iii) purple squares:
$J=\mu=\Delta$ and $\alpha=0$.}
\end{figure}

\subsection{HS to super-HS transition for estimating $\Delta$}

The maximum and minimum eigenvalues of $\partial_{\Delta}H$ in the
many-body Hilbert space are $\gamma_{\alpha}(N)/2$ and $0$, where
$\gamma_{\alpha}(N)\equiv\sum_{k}|f_{\alpha}(k)|$. Thus, for estimating
$\Delta$, the QFI reads 
\begin{equation}
I_{0}(\Delta)=[\gamma_{\alpha}(N)/2]^{2}T^{2}.\label{eq:I0-Delta}
\end{equation}
Determining the scaling of $I_{0}(\Delta)$ boils down to computing
the scaling of $\gamma_{\alpha}(N)$ at large $N$. Let us first discuss
a simple case, with $\kappa_{l,\,0}=1$. In this case, $f_{0}(k)=\cot(k/2)$.
According to Appendix~\ref{sec:Proof-Scaling-Gamma}, when applying
the Euler-Maclaurin formula, the upper bound of the scaling of the
remainder, which is the difference between $\gamma_{0}(N)$ and the
main integral $N/(2\pi)\int_{\pi/N}^{\pi}\cot(k/2)dk$, is $N$ due
to the singularity of $f_{0}(k)$ around $k=0$. Nevertheless, since
the main integral$N/(2\pi)\int_{\pi/N}^{\pi}\cot(k/2)dk\sim N\ln N$,
which is still much larger than $N$ in the asymptotic limit of larger
$N$, we conclude that the leading order $\gamma_{0}(N)$ is $N$
.

Let us next focus on the case of the power-law decay for the long-range
interaction in the original proposal of LRK~\citep{vodola2014kitaevchains},
i.e., $\kappa_{l,\,\alpha}=l^{-\alpha}$. For $\alpha>1$, $f_{\alpha}(k)$
has no singularity for all the values of the momentum $k$. This is
because $|\sum_{l=1}^{N-1}\sin(kl)/l^{\alpha}|\leq\sum_{l=1}^{N-1}1/l^{\alpha}$
and the latter series is convergent for $\alpha>1$. With the Euler-Maclaurin
formula discussed in Appendix~\ref{sec:Euler-Maclaurin-formula},
$\gamma_{\alpha}(N)$ scales as $N$. Thus the QFI $I_{0}(\Delta)$
obeys the HS for $\alpha>1$. For $0<\alpha\le1$, using the properties
of polylogarithmic function~\citep{olver2010nisthandbook}, one finds
$f_{\alpha}(k)\sim1/k^{1-\alpha}$ (see also Appendix~\ref{sec:Singularity-f-alpha-power}
and Ref.~\citep{vodola2014kitaevchains}). According to Appendix~\ref{sec:Proof-Scaling-Gamma},
the upper bound of the scaling of the remainder in the Euler-Maclaurin
formula is strictly slower than $N$ and the leading order of $\gamma_{\alpha}(N)$
is controlled by the main integral $N/2\pi\int_{\pi/N}^{\pi}|f_{\alpha}(k)|dk.$
We note that $\int_{\pi/N}^{\pi}[|f_{\alpha}(k)|-1/k^{1-\alpha}]dk$
should be constant as $N\to\infty$ since the singularity has been
removed. So $\int_{\pi/N}^{\pi}|f_{\alpha}(k)|dk\sim\int_{\pi/N}^{\pi}1/k^{1-\alpha}dk$
is a constant, which does not scale with $N$ and therefore $\gamma_{\alpha}(N)\sim N$.
We thus find that for $\kappa_{l,\,\alpha}=l^{-\alpha}$, super-HS
scaling only occurs for $\alpha=0$.

\textcolor{blue}{Now, let us explore more general long-range interactions
that satisfy the regularity condition at the beginning. As we have
seen above, the scaling $\gamma_{\alpha}(N)$ crucially depends on
the singularities of $f_{\alpha}(k)$, which is caused by the slow-decaying
long-range interactions. We argue at the end of Appendix~\ref{sec:The-convergence-integration-by-parts}
that $\int_{1/N}^{\Lambda}dkf_{\alpha}(k)\sim N\int_{1}^{N}(\kappa_{x,\,\alpha}/x)\le N\ln N$.
Then according to }Appendix~\ref{sec:Proof-Scaling-Gamma},\textcolor{blue}{{}
we find the leading order scaling of $\gamma_{\alpha}(N)$ is controlled
by $N/2\pi\int_{1/N}^{\Lambda}dkf_{\alpha}(k)\sim N\int_{1}^{N}(\kappa_{x,\,\alpha}/x)dx$.}
We see that the maximum possible scaling $\gamma_{\alpha}(N)$ is
$N\ln N$, where $\kappa_{x,\,\alpha}$ is a constant that does not
depend on $x$. Therefore, according to Eq.~\eqref{eq:I0-Delta},
\begin{equation}
I_{0}(\Delta)\sim N^{2}\left[\int_{1}^{N}(\kappa_{x,\,\alpha}/x)dx\right]^{2}.\label{eq:I0-Delta-scaling}
\end{equation}
and it is bounded by $N^{2}(\ln N)^{2}$ rather than the HS. In particular,
when the long-range interaction decays sufficiently slow, $\int_{1}^{N}(\kappa_{x,\,\alpha}/x)$
can diverge at large $N$ and therefore super-HS occurs for $I_{0}(\Delta)$.
This is the case, e.g., when $\kappa_{x,\,\alpha}=[\ln(ex)]^{-\alpha}=(1+\ln x)^{-\alpha}$
which satisfies the regularity conditions with $Q=1$. So we obtain
$\gamma_{\alpha}(N)\sim N\int_{1}^{N}dx/[x(1+\ln x)^{\alpha}]$. The
integral can be evaluated with the change of variable $s=1+\ln x$,
which leads to 
\begin{equation}
I_{0}(\Delta)\sim\begin{cases}
N^{2}(\ln N)^{2(1-\alpha)} & \alpha\in[0,\,1)\\
N^{2}(\ln\ln N)^{2} & \alpha=1\\
N^{2} & \alpha>1
\end{cases}.\label{eq:scaling0-double-int}
\end{equation}
As a result, super-HS occurs for the very slow decay law dictated
by the power of logarithms when $\alpha\le1$. \textcolor{blue}{As
one can see from Fig.~\ref{fig:The-plots-of-IDelta}~(a)-(b), the
analytical scalings of $I_{0}(\Delta)$ for $\kappa_{0,\,l}=1$ and
$\kappa_{0.2,\,l}=(1+\ln l)^{-0.2}$ shown by the blue solid lines,
are in excellent agreement with their respective numerical calculations,
shown by the cyan and red triangles in Fig.~\ref{fig:The-plots-of-IDelta}~(a)-(b),
respectively.}

\begin{figure}
\begin{picture}(200, 240) \put(5,130){\includegraphics[width=0.7\linewidth]{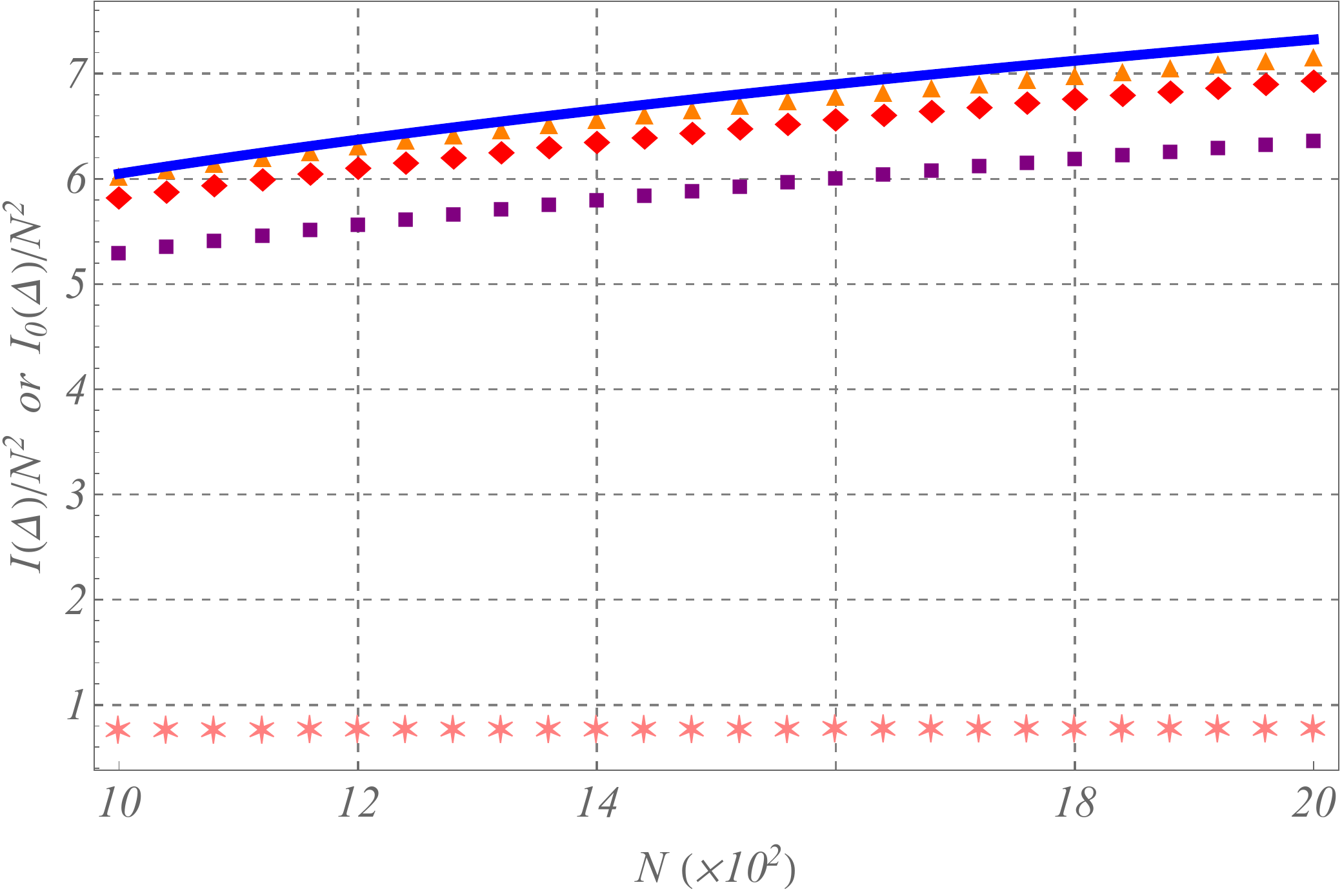}}\put(0,
0){\includegraphics[width=0.7\linewidth]{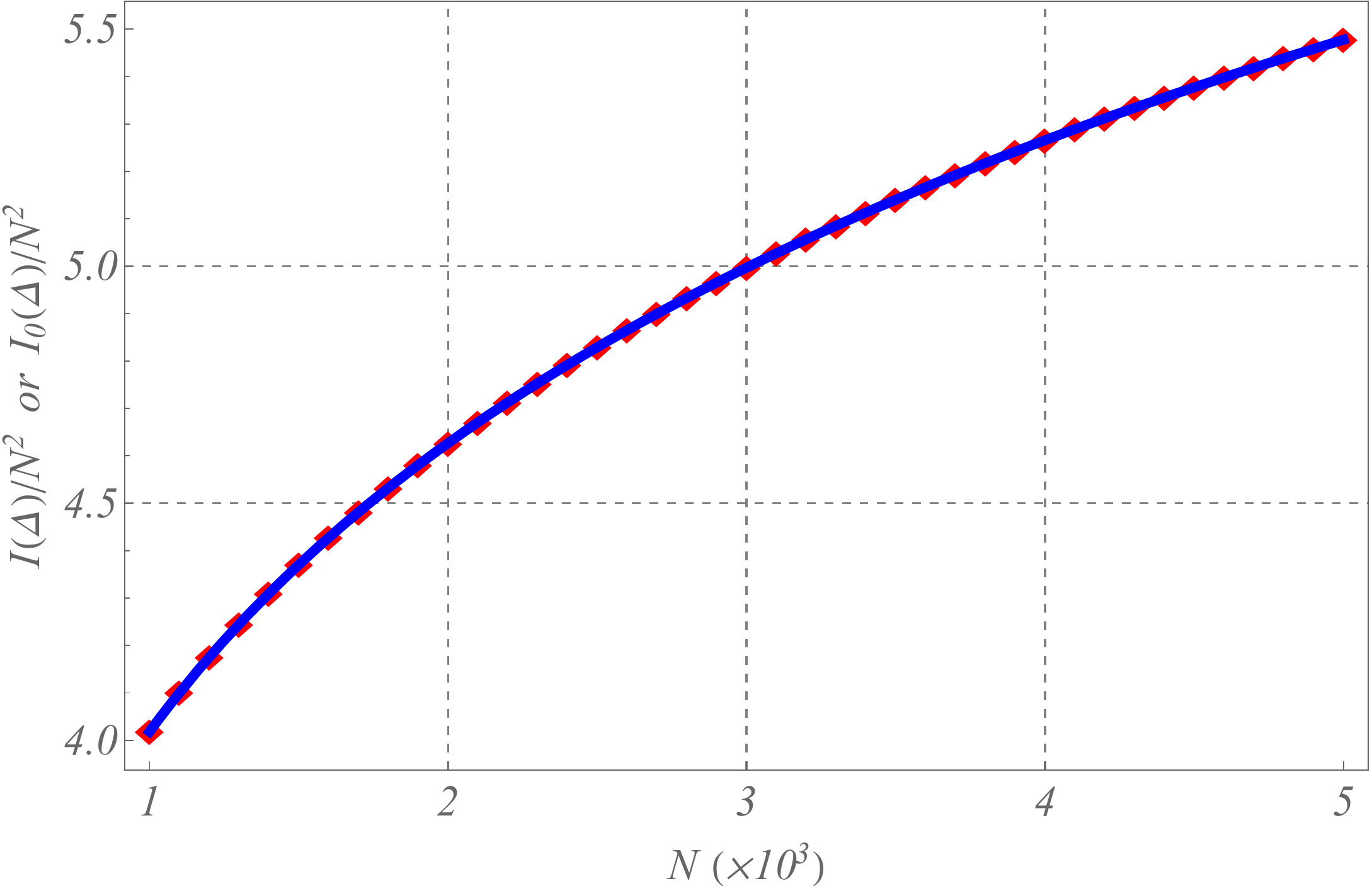}} \put(25,230){\footnotesize{}{}{}(a)}
\put(25,100){\footnotesize{}{}{}(b)} \end{picture}

\caption{\label{fig:The-plots-of-IDelta} Quantum Fisher information $I(\Delta)$
for $\Delta$ estimation in the case with no control or imperfect
control, as well as for optimal control as function of $N$ for (a)
$\kappa_{l,\,\alpha}=l^{-\alpha}$ and (b) $\kappa_{l,\,\alpha}=(1+\ln l)^{-\alpha}$.
The probe time $T=1$ in both figures. (a) All the discrete points
are numerically calculated from Eq.~(\ref{eq:I-theta}). The values
of the parameters for (i) red circle dots: $J=\mu=0.5\Delta$ and
$\alpha=0$ (ii) purple squares: $J=\mu=\Delta$ and $\alpha=0$ (iii)
pink stars are $J=\mu=\Delta$ and $\alpha=0.5$ (iv) Cyan triangles:
$J=\mu=0$ and $\alpha=0$. The blue solid line is the scaling $N^{2}(\ln N)^{2}$,
where the prefactor is determined by the QFI for $J=\mu=0$ and $N=1000$.
(b) The red triangles are numerical calculations of $I_{0}(\Delta)$
for $\alpha=0.2$ and the blue line is the fitted to the red triangles
with $A(\ln N)^{c}+B$, where $A=0.20$, $c=1.54$ and $B=0.17$.
The values $c$ is very close to the expected value $2(1-\alpha)=1.6$.
The slight deviation of the scaling exponent between theory and the
fitted results is because $\ln N$ is a very slowly increasing function
compared to the power functions.}
\end{figure}

\subsection{Resilience of the scaling under no or imperfect control}

We have seen that the HS of $I_{0}(J)$ and $I_{0}(\mu)$ is due to
the fact that the spectrum of $\mathscr{E}_{J}(k)$ and $\mathscr{E}_{\mu}(k)$
is regular near $k=0$, while the possibility of super-HS scaling
in $I_{0}(\Delta)$ is due to the fast divergence of $\mathscr{E}_{\Delta}(k)$
near $k=0$. It is natural to consider the fate of these scaling laws
when control is not optimally applied or is not available. According
to Eq.~(\ref{eq:I}), we find 
\begin{equation}
I(\theta)=\left[\sum_{k}\mathscr{E}_{\theta}(k)\right]^{2}.\label{eq:I-theta}
\end{equation}
Let us first discuss the estimation of $J$. First, from Eqs.~(\ref{eq:eps-theta},~\ref{eq:xi-k-def}),
one can readily obtain
\begin{equation}
\partial_{J}\epsilon_{J}(k)=\cos k(J\cos k+\mu)/\epsilon_{J}(k),
\end{equation}
\begin{equation}
\xi_{J}(k)=\Delta f_{\alpha}(k)\cos k/[2\epsilon_{J}^{2}(k)].
\end{equation}
Since we focus on the no-control or imperfect control case, $\Delta\neq0$.
We see that the only possibility for $\partial_{J}\epsilon_{J}(k)$
and $\partial_{\mu}\epsilon_{\mu}(k)$ to blow up is when their denominators
vanish, i.e., $J\cos k+\mu=0$ and $f_{\alpha}(k)=0$ near $k=0$.
However, we note that whenever $f_{\alpha}(k)=0$, $\partial_{J}\epsilon_{J}(k)=\pm\cos k$.
The same argument also applies to $\xi_{J}(k)$. Therefore $\mathscr{E}_{J}(k)$
does not blow up. Thus we conclude in the absence of controls or in
the presence of imperfect control, the HS is not affected. For the
estimation of $\Delta$, it is readily found from Eqs.~(\ref{eq:eps-theta},~\ref{eq:xi-k-def})
that
\begin{equation}
\partial_{\Delta}\epsilon_{\Delta}(k)=\Delta f_{\alpha}^{2}(k)/[4\epsilon_{\Delta}(k)],
\end{equation}
\begin{equation}
\xi_{\Delta}(k)=-(J\cos k+\mu)f_{\alpha}(k)/[2\epsilon_{\Delta}^{2}(k)].
\end{equation}
Since around $k=0$, $\epsilon_{\Delta}(k)\sim f_{\alpha}(k)$, we
know $\lim_{k\to0}\xi_{\Delta}(k)=0$ and $\partial_{\Delta}\epsilon_{\Delta}(k)\sim f_{\alpha}(k)$.
Therefore, the dominant divergence in $\mathscr{E}_{\Delta}(k)$ is
controlled by $\partial_{\Delta}\epsilon_{\Delta}(k)$ and is the
same as the case of the optimal estimation of $\Delta$. The scaling
of estimating $\Delta$ is again unchanged.

\textcolor{blue}{Fig.~\ref{fig:The-plots-of-IJ} and \ref{fig:The-plots-of-IDelta}
show a comparison between the cases with optimal controls and with
no controls or imperfect controls, respectively. As we can see from
these figures, the slopes of the lines for the cases with no or imperfect
control match the one for optimal control. The same conclusion holds
for the estimation of $\mu$. Therefore, the role of optimal quantum
controls here is to improve the prefactor of the leading order scaling
of the ultimate QFI rather than the scaling exponent. }

\section{Discussions and conclusions}

We have established that the scaling of the QFI for estimating the
superconducting strength $\Delta$ is bounded by $N^{2}(\ln N)^{2}$
rather than the HS due to the long-range interactions. As in Eq.~\eqref{eq:LRK},
long-range interactions contains $N^{2}$ terms whose strength is
controlled by $\kappa_{l,\,\alpha}$. Intuitively, if $\kappa_{l,\,\alpha}$
decays quickly, these $N^{2}$ terms effectively behave like a local
interaction containing only $N$ terms, like in the estimation of
$J$ and $\mu$, and lead to the HS in estimating $\Delta$. However,
if $\kappa_{l,\,\alpha}$ decays sufficiently slow, these $N^{2}$
terms can collectively give rise to the super-HS behavior. We have
illustrated this in two examples with $\kappa_{l,\,\alpha}=l^{-\alpha}$
and $\kappa_{l,\,\alpha}=(1+\ln l)^{-\alpha}$, respectively. Interestingly,
when $N$ is not large enough, we have shown in Appendix~\ref{sec:Finite-size-scaling}
that super-HS $N^{2}(\ln N)^{2}$ and $N^{2}(\ln\ln N)^{2}$ can also
occur as long as 
\begin{equation}
\epsilon\ll(\ln N)^{-1}
\end{equation}
 for $\kappa_{x,\,\epsilon}=x^{-\epsilon}$ and 
\begin{equation}
\epsilon\ll\ln\ln N
\end{equation}
for $\kappa_{l,\,1+\epsilon}=(1+\ln l)^{-(1+\epsilon)}$, respectively.
Note that the LRK model here is \textit{linear}, and thus different
from the super-HS in the nonlinear models ~\citep{beau2017nonlinear,boixo2007generalized,roy2008exponentially}.
Since the HS characterizes the many-body entanglement of the probes
if the generator only contains local operators~\citep{hyllus2012fisherinformation,toth2012multipartite,pezz`e2017multipartite},
our results may indicate there may be an intimate connection between
the HS to super-HS transition and the property of quantum entanglement.

One can also view the super-HS in the spin representation: LRK Hamiltonian~(\ref{eq:LRK})
can be transformed into the one for spin systems via the Jordan-Wigner
transformation~\citep{coleman2015introduction}. The resulting Hamiltonian
becomes (see Appendix~\ref{sec:spin-rep})
\begin{align}
H_{\text{spin}} & =-\frac{J}{4}\sum_{j=1}^{N}(\sigma_{j}^{x}\sigma_{j+l}^{x}+\sigma_{j}^{y}\sigma_{j+l}^{y})-\frac{\mu}{2}\sum_{j=1}^{N}\sigma_{j}^{z}\nonumber \\
 & -\frac{\Delta}{8}\sum_{j=1}^{N}\kappa_{1,\,\alpha}(\sigma_{j}^{x}\sigma_{j+l}^{x}-\sigma_{j}^{y}\sigma_{j+l}^{y})\nonumber \\
 & +\frac{\Delta}{8}\sum_{j=1}^{N}\sum_{l=2}^{N-1}(-1)^{l}\kappa_{l,\,\alpha}(\sigma_{j}^{x}\sigma_{j+l}^{x}-\sigma_{j}^{y}\sigma_{j+l}^{y})\otimes_{k=1}^{l-1}\sigma_{j+k}^{z},\label{eq:Hspin}
\end{align}
which contains the long-range pairing term involves interaction among
$(l+1)$-spins, with $1\le l\le N-1$. This agrees with the intuition
that for spin systems, reaching the super-HS requires interactions
involving more than one single spin operator~\citep{roy2008exponentially,boixo2007generalized}.

We have further shown that the singularity is not altered by whether
external control is optimally applied or not. Therefore, we conclude
that in the LRK model, quantum controls can improve the ultimate QFI
by altering the prefactor while preserving the scaling exponent.

\textcolor{blue}{Our results are of direct relevance to practical
quantum metrology with quantum dots~\citep{qiao2021longdistance},
trapped ions~\citep{richerme2014nonlocal,jurcevic2014quasiparticle}
and cold atoms~\citep{hung2016quantum}. } Our findings should be
applicable to the relation between the HS /super-HS and the many-body
entanglement~\citep{hyllus2012fisherinformation,toth2012multipartite,pezz`e2017multipartite},
the physical preparations of the optimal initial states (Fermionic
GHZ states)~\citep{shapourian2019entanglement}, optimal detection
associated with the HS and super-HS~\citep{yang2019optimal,braunstein_statistical_1994,zhou2020saturating}
and quantum estimation of the LRK in the presence of decoherence and
dissipation~\citep{beau2017nonlinear}.

\section{Acknowledgement}

We thank Hongzhe Zhou for useful discussions. Part of this work was
done when JY visited SP at Sun Yat-Sen University (SYSU),~China,
in May 2021. We would like to thank SYSU for their warmth and hospitality.
Support from the National Natural Science Foundation of China (NSFC)
Grants No.~12075323, the NSF Grants No.~DMR-1809343, and US Army
Research Office Grants No.~W911NF-18-10178 is greatly acknowledged.

\appendix

\section{\label{sec:The-diagonaization}The diagonalization of the LRK Hamiltonian
and the generator for parameter estimation}

We note that the LKR Hamiltonian in momentum space reads~\citep{pezz`e2017multipartite}
\begin{widetext}
\begin{align}
H_{\theta} & =-\frac{J}{2}\sum_{k}\left[a^{\dagger}(k)a(k)+a^{\dagger}(-k)a(-k)\right]\cos k-\frac{\mu}{2}\left[\sum_{k}a^{\dagger}(k)a(k)+\sum_{-k}a^{\dagger}(-k)a(-k)\right]\nonumber \\
 & +\frac{\text{i}\Delta}{4}\sum_{k}\left[a(-k)a(-k)-a^{\dagger}(-k)a^{\dagger}(-k)\right]f_{\alpha}(k).\label{eq:LRK-p-space}
\end{align}
\end{widetext}The Fourier transformation that relates the original
Hamiltonian~(\ref{eq:LRK}) to above Hamiltonian does not depend
on any estimation parameters. Thus, the Fisher information is preserved
by the transformation. Eq.~(\ref{eq:LRK-p-space}) can be diagonalized
as follows

\begin{equation}
H_{\theta}=\frac{1}{2}\sum_{k}\epsilon_{\theta}(k)\begin{bmatrix}a^{\dagger}(k),\, & a(-k)\end{bmatrix}U_{\theta}^{\dagger}(k)\sigma_{z}U_{\theta}(k)\begin{bmatrix}a(k)\\
a^{\dagger}(-k)
\end{bmatrix},
\end{equation}
through the Bogoliubov transformation $U_{\theta}(k)$ 
\begin{align}
U_{\theta}(k) & \equiv\begin{pmatrix}\cos\left[\phi_{\theta}(k)/2\right] & \text{i}\sin\left[\phi_{\theta}(k)/2\right]\\
\text{i}\sin\left[\phi_{\theta}(k)/2\right] & \cos\left[\phi_{\theta}(k)/2\right]
\end{pmatrix},\\
\sin\phi_{\theta}(k) & =-\frac{\Delta}{2}\frac{f_{\alpha}(k)}{\epsilon_{\theta}(k)},\\
\cos\phi_{\theta}(k) & =-\frac{(J\cos k+\mu)}{\epsilon_{\theta}(k)}.
\end{align}
Denoting 
\begin{equation}
\begin{bmatrix}\eta_{\theta}(k)\\
\eta_{\theta}^{\dagger}(-k)
\end{bmatrix}\equiv U_{\theta}(k)\begin{bmatrix}a(k)\\
a^{\dagger}(-k)
\end{bmatrix},
\end{equation}
the Hamiltonian can be rewritten as 
\begin{equation}
H_{\theta}=\sum_{k}\epsilon_{\theta}(k)\left[\eta_{\theta}^{\dagger}(k)\eta_{\theta}(k)-\frac{1}{2}\right].\label{eq:H-diagonal}
\end{equation}
The parameter-dependent constant $\sum_{k}\epsilon_{\theta}(k)/2$
does not contribute the Fisher information and will be suppressed.
The generator for parameter estimation is~\citep{pang2017optimal,pang2014quantum}

\begin{equation}
G_{\theta}=\int_{0}^{T}\mathcal{U}_{\theta}^{\dagger}(\tau)\partial_{\theta}H_{\theta}\mathcal{U}_{\theta}(\tau)d\tau,\label{eq:G-def}
\end{equation}
where the evolution operator is $\mathcal{U}_{\theta}(\tau)=e^{-\text{i}H_{\theta}\tau}$.
According to Eq.~(\ref{eq:H-diagonal}), the generator contains two
parts: The first part is due to the $\partial_{\theta}\epsilon_{\theta}(k)$
and the other is due to $\partial_{\theta}\eta_{\theta}^{\dagger}$
and $\partial_{\theta}\eta_{\theta}(k)$. It is readily found that
\begin{gather}
\begin{bmatrix}\partial_{\theta}\eta_{\theta}(k)\\
\partial_{\theta}\eta_{\theta}^{\dagger}(-k)
\end{bmatrix}=\frac{d\phi_{\theta}}{d\theta}[\partial_{\phi}U_{\theta}(k)U_{\theta}^{-1}(k)]\begin{bmatrix}\eta_{\theta}(k)\\
\eta_{\theta}^{\dagger}(-k)
\end{bmatrix}\nonumber \\
=\frac{\text{i}}{2}\frac{d\phi_{\theta}(k)}{d\theta}\sigma_{x}\begin{bmatrix}\eta_{\theta}(k)\\
\eta_{\theta}^{\dagger}(-k)
\end{bmatrix}=-\frac{\text{i}\partial_{\theta}\cos\phi_{\theta}(k)}{2\sin\phi_{\theta}(k)}\sigma_{x}\begin{bmatrix}\eta_{\theta}(k)\\
\eta_{\theta}^{\dagger}(-k)
\end{bmatrix}.
\end{gather}
Therefore, \begin{widetext}
\begin{equation}
\partial_{\theta}H_{\theta}=\sum_{k}\partial_{\theta}\epsilon_{\theta}(k)\eta_{\theta}^{\dagger}(k)\eta_{\theta}(k)+\frac{\text{i}}{2}\sum_{k}\xi_{\theta}(k)\epsilon_{\theta}(k)\left(\eta_{\theta}(-k)\eta_{\theta}(k)-\eta_{\theta}^{\dagger}(k)\eta_{\theta}^{\dagger}(-k)\right).\label{eq:dHdtheta}
\end{equation}
Substituting Eq.~(\ref{eq:dHdtheta}) into Eq.~(\ref{eq:G-def}),
we find 
\begin{equation}
G_{\theta}=T\sum_{k}\partial_{\theta}\epsilon_{\theta}(k)\eta_{\theta}^{\dagger}(k)\eta_{\theta}(k)+\frac{\text{i}}{2}\sum_{k}\xi_{\theta}(k)\epsilon_{\theta}(k)\int_{0}^{T}d\tau e^{\text{i}\sum_{p}\epsilon_{\theta}(p)\eta_{\theta}^{\dagger}(p)\eta_{\theta}(p)\tau}\left(\eta_{\theta}(-k)\eta_{\theta}(k)-\eta_{\theta}^{\dagger}(k)\eta_{\theta}^{\dagger}(-k)\right)e^{-\text{i}\sum_{p^{\prime}}\epsilon_{\theta}(p^{\prime})\eta_{\theta}^{\dagger}(p^{\prime})\eta_{\theta}(p^{\prime})\tau}.\label{eq:G}
\end{equation}
Note that the product of two Fermionic creation and annihilation operators
behaves like a $c$-number when commuting with Fermionic operators
in other modes. We further note that the negative modes are equivalent
to the positive modes via the identification $-k\sim2\pi-k$. With
these two observations, it is readily checked that

\begin{align}
 & e^{\text{i}\sum_{p}\epsilon_{\theta}(p)\eta_{\theta}^{\dagger}(p)\eta_{\theta}(p)\tau}\left(\eta_{\theta}(-k)\eta_{\theta}(k)-\eta_{\theta}^{\dagger}(k)\eta_{\theta}^{\dagger}(-k)\right)e^{-\text{i}\sum_{p^{\prime}}\epsilon_{\theta}(p^{\prime})\eta_{\theta}^{\dagger}(p^{\prime})\eta_{\theta}(p^{\prime})\tau}\nonumber \\
= & e^{\text{i}[\epsilon_{\theta}(k)\eta_{\theta}^{\dagger}(k)\eta_{\theta}(k)+\epsilon_{\theta}(2\pi-k)\eta_{\theta}^{\dagger}(2\pi-k)\eta_{\theta}(2\pi-k)]\tau}\left(\eta_{\theta}(-k)\eta_{\theta}(k)-\eta_{\theta}^{\dagger}(k)\eta_{\theta}^{\dagger}(-k)\right)e^{-\text{i}[\epsilon_{\theta}(k)\eta_{\theta}^{\dagger}(k)\eta_{\theta}(k)+\epsilon_{\theta}(2\pi-k)\eta_{\theta}^{\dagger}(2\pi-k)\eta_{\theta}(2\pi-k)]\tau}\nonumber \\
= & e^{\text{i}\epsilon_{\theta}(k)\tau[\eta_{\theta}^{\dagger}(k)\eta_{\theta}(k)+\eta_{\theta}^{\dagger}(-k)\eta_{\theta}(-k)]}\left(\eta_{\theta}(-k)\eta_{\theta}(k)-\eta_{\theta}^{\dagger}(k)\eta_{\theta}^{\dagger}(-k)\right)e^{-\text{i}\epsilon_{\theta}(k)\tau[\eta_{\theta}^{\dagger}(k)\eta_{\theta}(k)+\eta_{\theta}^{\dagger}(-k)\eta_{\theta}(-k)]}.\label{eq:exp-fermionic}
\end{align}
Using the relation 
\begin{equation}
e^{i\theta\eta^{\dagger}\eta}=1+\eta^{\dagger}\eta(e^{\text{i}\theta}-1)=e^{\text{i\ensuremath{\theta}}}-\eta\eta^{\dagger}(e^{\text{i}\theta}-1)
\end{equation}
for Fermionic operators, Eq.~(\ref{eq:exp-fermionic}) becomes 
\begin{align}
 & e^{\text{i}\epsilon_{\theta}(k)\tau[\eta_{\theta}^{\dagger}(k)\eta_{\theta}(k)+\eta_{\theta}^{\dagger}(-k)\eta_{\theta}(-k)]}\eta_{\theta}(-k)\eta_{\theta}(k)e^{-\text{i}\epsilon_{\theta}(k)\tau[\eta_{\theta}^{\dagger}(k)\eta_{\theta}(k)+\eta_{\theta}^{\dagger}(-k)\eta_{\theta}(-k)]}\nonumber \\
= & [1+\eta_{\theta}^{\dagger}(k)\eta_{\theta}(k)(e^{\text{i}\epsilon_{\theta}(k)\tau}-1)][1+\eta_{\theta}^{\dagger}(-k)\eta_{\theta}(-k)(e^{\text{i}\epsilon_{\theta}(k)\tau}-1)]\eta_{\theta}(-k)\eta_{\theta}(k)\nonumber \\
\times & [e^{-\text{i}\epsilon_{\theta}(k)\tau}-\eta_{\theta}(k)\eta_{\theta}^{\dagger}(k)(e^{-\text{i}\epsilon_{\theta}(k)\tau}-1)][e^{-\text{i}\epsilon_{\theta}(k)\tau}-\eta_{\theta}(-k)\eta_{\theta}^{\dagger}(-k)(e^{-\text{i}\epsilon_{\theta}(k)\tau}-1)]\nonumber \\
= & e^{-\text{i}\epsilon_{\theta}(k)\tau}[1+\eta_{\theta}^{\dagger}(k)\eta_{\theta}(k)(e^{\text{i}\epsilon_{\theta}(k)\tau}-1)]\eta_{\theta}(-k)\eta_{\theta}(k)[e^{-\text{i}\epsilon_{\theta}(k)\tau}-\eta_{\theta}(-k)\eta_{\theta}^{\dagger}(-k)(e^{-\text{i}\epsilon_{\theta}(k)\tau}-1)]\nonumber \\
= & -e^{-\text{i}\epsilon_{\theta}(k)\tau}[1+\eta_{\theta}^{\dagger}(k)\eta_{\theta}(k)(e^{\text{i}\epsilon_{\theta}(k)\tau}-1)]\eta_{\theta}(k)\eta_{\theta}(-k)[e^{-\text{i}\epsilon_{\theta}(k)\tau}-\eta_{\theta}(-k)\eta_{\theta}^{\dagger}(-k)(e^{-\text{i}\epsilon_{\theta}(k)\tau}-1)]\nonumber \\
= & -e^{-2\text{i}\epsilon_{\theta}(k)\tau}\eta_{\theta}(k)\eta_{\theta}(-k)=e^{-2\text{i}\epsilon_{\theta}(k)\tau}\eta_{\theta}(-k)\eta_{\theta}(k),
\end{align}
where we have used $\eta_{\theta}^{2}(-k)=\eta_{\theta}^{2}(k)=0$.
The generator now becomes 
\begin{align}
G_{\theta} & =\sum_{k}\left\{ T\partial_{\theta}\epsilon_{\theta}(k)\eta_{\theta}^{\dagger}(k)\eta_{\theta}(k)\right.\nonumber \\
 & \left.+\frac{\xi_{\theta}(k)}{4}\left[(1-e^{-2\text{i}\epsilon_{\theta}(k)T})\eta_{\theta}(-k)\eta_{\theta}(k)+(1-e^{2\text{i}\epsilon_{\theta}(k)T}))\eta_{\theta}^{\dagger}(k)\eta_{\theta}^{\dagger}(-k)\right]\right\} .\label{eq:G-general}
\end{align}
\end{widetext}We rewrite Eq.~\eqref{eq:G-general} in a more compact
form 
\begin{equation}
G_{\theta}=\frac{1}{2}\sum_{k}\begin{bmatrix}\eta_{\theta}^{\dagger}(k),\, & \eta_{\theta}(-k)\end{bmatrix}\mathscr{G}_{\theta}(k)\begin{bmatrix}\eta_{\theta}(k)\\
\eta_{\theta}^{\dagger}(-k)
\end{bmatrix},
\end{equation}
where the matrix $\mathscr{G}_{\theta}(k)$ is defined as 
\begin{align}
\mathscr{G}_{\theta}(k) & \equiv T\partial_{\theta}\epsilon_{\theta}(k)\sigma_{z}+\frac{\xi_{\theta}(k)}{2}\left(1-\cos[2\epsilon_{\theta}(k)T]\right)\sigma_{x}\nonumber \\
 & +\frac{\xi_{\theta}(k)}{2}\sin[2\epsilon_{\theta}(k)T]\sigma_{y}\nonumber \\
 & =\mathscr{E}_{\theta}(k)\bm{n}_{\theta}(k)\cdot\bm{\sigma}.
\end{align}
Furthermore we note that
\begin{equation}
V_{\theta}^{\dagger}(k)\left[\bm{n}_{\theta}(k)\cdot\bm{\sigma}\right]V_{\theta}(k)=\sigma_{z},
\end{equation}
where 
\begin{equation}
V_{\theta}(k)=\begin{pmatrix}\ket{\uparrow_{\bm{n}_{\theta}(k)}},\, & \ket{\downarrow_{\bm{n}_{\theta}(k)}}\end{pmatrix}.
\end{equation}
and $\ket{\uparrow_{\bm{n}_{\theta}(k)}}$ and $\ket{\downarrow_{\bm{n}_{\theta}(k)}}$
are the vectors aligned and anti-aligned with the vector $\bm{n}_{\theta}(k)$
on the Bloch sphere, respectively. Introducing 
\begin{equation}
\begin{bmatrix}\psi_{\theta}(k)\\
\psi^{\dagger}(-k)
\end{bmatrix}\equiv V_{\theta}(k)\begin{bmatrix}\eta_{\theta}(k)\\
\eta^{\dagger}(-k)
\end{bmatrix},\label{eq:psi-k-def}
\end{equation}
one can readily obtain Eqs.~(\ref{eq:Generator-diagonal},~\ref{eq:calEk})
in the main text.

\section{\label{sec:Euler-Maclaurin-formula}The Euler-Maclaurin formula}
\begin{lem}
\label{lem:series-integral-rep}(Euler-Maclaurin formula) For arbitrary
function $g(x)$ with continuous derivatives, the infinite series
$\sum_{n=a}^{b}g(m)$ can be converted the corresponding integral
plus remainder terms via the Euler-Maclaurin formula~\citep{knopp1990theoryand}.
\begin{equation}
\sum_{n=a}^{n=b}g(n)=\int_{a}^{b}g(x)dx+R,\label{eq:E-M}
\end{equation}
where the remainder is 
\begin{align}
R & =\frac{1}{2}[g(b)-g(a)]+\sum_{m=1}^{M}\frac{b_{2m}}{(2m)!}\left[g^{(2m-1)}(b)-g^{(2m-1)}(a)\right]\nonumber \\
 & +\int_{a}^{b}\frac{1}{(2M+1)!}P_{2M+1}(x)g^{(2M+1)}(x)dx.
\end{align}
Here, $M$ can be arbitrarily chosen from the natural numbers $0,1,\,2,\cdots$,
$b_{2m}$ is the Bernoulli number. $P_{0}(x)=1$ for $M>0$ 
\begin{equation}
P_{M}(x)=\frac{1}{M!}B_{M}(\{x\}),\label{eq:P-def}
\end{equation}
where$\{x\}\equiv x-[x]$ and $B_{M}$ is the Bernoulli polynomial.
\end{lem}

We can use the Euler-Maclaurin formula to approximate a series 
\begin{equation}
\sum_{n=a}^{b}f\left(\frac{(2n+1)\pi}{N}\right)=\sum_{k=k_{a}}^{k_{b}}f(k),
\end{equation}
where we $k=(2n+1)\pi/N$. We assume $f(k)$ is piecewise smooth on
$[k_{a},\,k_{b}]$ and does not blow up on $[k_{a},\,k_{b}]$. \textit{We
allow some discontinuities in the first derivatives if $M=1$ so that
$f(k)$ may contain an absolute value or a square root}. Without loss
of generality, we can assume $f(k)$ is smooth on interval $F_{j}$'s
where $\cup_{j}F_{j}=[k_{a},\,k_{b}]$. When denoting the function
in terms of the variable , these intervals are denotes as $E_{j}$'s.
Applying the Euler-Maclaurin formula for these intervals respectively
with $M=0$, we find

\begin{equation}
\sum_{n=0}^{N-1}f\left[\frac{(2n+1)\pi}{N}\right]=\sum_{j}\left\{ \int_{E_{j}}f\left[\frac{(2x+1)\pi}{N}\right]dx+R_{j}\right\} ,
\end{equation}
where 
\begin{equation}
R_{j}=\frac{2\pi}{N}\int_{E_{j}}P_{1}(x)\sin\left[\frac{(2x+1)\pi}{N}\right]dx+\text{boundary terms},
\end{equation}
and $P_{1}(x)$ is defined in Eq.~(\ref{eq:P-def}). We note that
the boundary terms remains finite and does not scale with $N$. They
will be omitted subsequently. In the Fourier representation, we find
\begin{align}
\sum_{n=0}^{N-1}f\left(\frac{[2n+1]\pi}{N}\right) & =\frac{N}{2\pi}\int_{\pi/N}^{2\pi-\pi/N}f(k)dk+\sum_{j}R_{j},
\end{align}
where 
\begin{equation}
R_{j}=\int_{F_{j}}P_{1}\left(\frac{Nk}{2\pi}-\frac{1}{2}\right)f^{\prime}(k)dk.
\end{equation}
Since $f(k)$ is differentiable on $F_{j}$, $f^{\prime}(k)$ is regular
on $F_{j}$. On the other hand, $P_{1}(x)\in[-1/2,\,1/2]$ is bounded.
We find that $R_{j}$ remains finite as long as the number of the
$F_{j}$'s does not scale with $N$. \textit{So we conclude that when
$f(k)$ is regular, $\sum_{k=k_{a}}^{k=k_{b}}f(k)\sim N$.} For example,
in Eq.~(\ref{eq:I0-J}) of the main text, we take $f(k)=|\cos k|$,
which is differentiable on $F_{1}=[\pi/N,\,\pi/2-\pi/N]$, $F_{2}=[\pi/2+\pi/N,\,3\pi/2-\pi/N]$
and $F_{3}=[3\pi/2+\pi/N,\,2\pi-\pi/N]$ respectively.

However, we note that if $f(k)$ has a singularity in $[0,\,2\pi]$,
the remainder may not be necessarily stay as a constant as $N\to\infty$.
For example, if we take 
\begin{equation}
f(k)=\cot\left[\frac{k}{2}\right],\label{eq:cot}
\end{equation}
where $k\in[\pi/N,\,\pi-\pi/N]$. Then we obtain

\begin{equation}
\sum_{n=0}^{N/2-1}\cot\left[\frac{(2n+1)\pi}{2N}\right]=\frac{N}{\pi}\int_{\pi/N}^{\pi-\pi/N}\cot\left(\frac{k}{2}\right)dk+R,\label{eq:gamma0-EM}
\end{equation}
where 
\begin{align}
R & =\frac{1}{2}\left[\cot\left(\frac{\pi}{2}-\frac{\pi}{2N}\right)-\cot\left(\frac{\pi}{2N}\right)\right]\nonumber \\
 & -\int_{\pi/N}^{\pi-\pi/N}P_{1}\left(\frac{Nk}{2\pi}-\frac{1}{2}\right)\frac{1}{\sin^{2}(k/2)}dk,
\end{align}
with $P_{1}(t)$ given in Eq.~(\ref{eq:P-def}). The integrand in
the remainder has a singularity around $k=0$ and there main integral
is no longer a good approximation of the sum. Nevertheless we can
upper bound the scaling of the integral in the remainder, i.e., 
\begin{align}
 & \bigg|\int_{\pi/N}^{\pi-\pi/N}P_{1}\left(\frac{Nk}{2\pi}-\frac{1}{2}\right)\frac{1}{\sin^{2}(k/2)}dk\bigg|\nonumber \\
\le & \int_{\pi/N}^{\pi-\pi/N}\frac{1}{\sin^{2}(k/2)}\sim\cot\left(\frac{\pi}{2N}\right).\label{eq:R0-bound}
\end{align}
We thus conclude the remainder will scale at most as $N$. Since $\int_{\pi/N}^{\pi-\pi/N}\cot(k/2)dk\sim N\ln N$,
we obtain the scaling of $\gamma_{0}(N)$ in the main text. We see
that in the current case\textit{ }the remainder depends on $N$ instead
of being a constant as indicated in Eq.~(\ref{eq:R0-bound}). We
would like to emphasize that \textit{when the summand of a sum has
a singularity in the limit $N\to\infty$, it is not rigorous to analyze
the scaling of the sum only with the main integral because the remainder
may contribute to the scaling. }

\section{\label{sec:Proof-Scaling-Gamma} The scaling of $\gamma_{\alpha}(N)$
for $f_{\alpha}(k)\le\mathcal{O}(1/k)$ near $k=0$}
\begin{thm}
We shall assume the only possible singularity of $f_{\alpha}(k)$
is near $k=0$, a fact which we will prove in Corollary~\ref{cor:sing-at-zero}.
Then the scaling of $\gamma_{\alpha}(N)$ is controlled by the main
integral if $f_{\alpha}(k)\le\mathcal{O}(1/k)$ near $k=0$.
\end{thm}

\begin{proof}
Let us first focus on the case $f_{\alpha}(k)$ is strictly slower
than $1/k$ near $k=0$. We denote $E_{j}$ as the intervals where
$f_{\alpha}([2x+1]\pi/N)$ is smooth as function $x$. Similar as
Sec.~\ref{sec:Euler-Maclaurin-formula}, this denomination allows
$f_{\alpha}(k)$ to be piecewise functions joined by smooth functions,
as long as there are no singularities at the joints. The intervals
$E_{j}$becomes $F_{j}$ when the function is written in terms of
the variable $k$. In particular, one can easily show that $f_{\alpha}(\pi/N)$
is positive. Applying Euler-Maclaurin formula~(\ref{eq:E-M}) to
each of these intervals, we find 
\begin{equation}
\gamma_{\alpha}(N)=2\sum_{j}\left[\int_{E_{j}}(-1)^{j-1}f_{\alpha}\left(\frac{[2x+1]\pi}{N}\right)dx+R_{\alpha j}\right],
\end{equation}
where the remainder is 
\begin{equation}
R_{\alpha j}=\int_{E_{j}}P_{1}\left(x\right)(-1)^{j-1}f^{\prime}\left(\frac{[2n+1]x\pi}{N}\right)dx+\text{boundary terms}.
\end{equation}
Since the boundary terms does not scale with $N$, we shall suppress
them in subsequent analysis. Now we change $x$ back to $k$, we find
\begin{equation}
\gamma_{\alpha}(N)=2\left[\frac{N}{2\pi}\int_{\pi/N}^{\pi-\pi/N}\big|f_{\alpha}(k)\big|dk+\sum_{j}R_{\alpha j}\right],\label{eq:gamma-EM}
\end{equation}
where 
\begin{equation}
R_{\alpha j}=\int_{F_{j}}P_{1}\left(\frac{Nk}{2\pi}-\frac{1}{2}\right)(-1)^{j-1}f_{\alpha}^{\prime}(k)dk.
\end{equation}
For the remainders, if $F_{j}$ does not contain the origin, then
the integral in $R_{j}$ is regular and does not scale with constant.
For $F_{j}$ contains the origin, we use a common trick in asymptotic
analysis~\citep{bender2013advanced}: The leading order of a singular
integral can be found by replacing the integrand with its leading
order Laurent expansion near the singular point. In our current case,
since 
\begin{equation}
f_{\alpha}(k)<\mathcal{O}\left(\frac{1}{k}\right),\label{eq:f-singu}
\end{equation}
we find 
\begin{align}
 & \bigg|\int_{F_{j}}P_{1}\left(\frac{Nk}{2\pi}-\frac{1}{2}\right)(-1)^{j-1}f_{\alpha}^{\prime}(k)dk\bigg|\nonumber \\
< & \bigg|\int_{F_{j}}P_{1}\left(\frac{Nk}{2\pi}-\frac{1}{2}\right)(-1)^{j-1}\left(\frac{1}{k}\right)^{\prime}dk\bigg|\nonumber \\
\leq & \bigg|\int_{\pi/N}\left(\frac{1}{k}\right)^{\prime}dk\bigg|\sim N,\label{eq:R-bound}
\end{align}
where we have used that $P_{1}(x)\in[-1/2,\,1/2]$ is bounded. That
is, the remainder scale scale strictly slower than $N$, which is
subleading order compared to the first term on the r.h.s. of Eq.~(\ref{eq:gamma-EM}).

When $f_{\alpha}(k)\sim\mathcal{O}(1/k)$, one can go through the
same argument and will find that the main integral will scale as $N\ln N$
while the upper bound of the scaling of the remainder is $N$. Therefore,
we conclude that the leading order scaling of $\gamma_{\alpha}(N)$
is only given by the main integral if $f_{\alpha}(k)\le\mathcal{O}(1/k)$
near $k=0$.
\end{proof}
We conclude this section by note that the condition $f_{\alpha}(k)\le\mathcal{O}(1/k)$
is non-trivial and essential: \textit{Had $f_{\alpha}(k)$ scaled
as $1/k^{1+\varepsilon}$ near $k=0$, where $\varepsilon$ is an
arbitrary positive number, the above proof would yield that both the
main integral and the upper bound of the remainder $R_{\alpha j}$
scales $N^{1+\varepsilon}$}. The analysis of the scaling of $\gamma_{\alpha}(N)$
would be subtle because the leading order scaling of the main integral
and the remainder $R_{\alpha j}$ might cancel each other. Fortunately,
we see such a situation does not occur because we have shown in the
main text that $f_{\alpha}(k)\le\mathcal{O}(1/k)$.

\section{\label{sec:Singularity-f-alpha-power}The singularity of $f_{\alpha}(k)$
for $\kappa_{l,\,\alpha}=l^{-\alpha}$ with $\alpha\in(0,\,1]$}

In this section, we prove an analytic property of $f_{\alpha}(k)$
for the particular case where $\kappa_{l,\,\alpha}=l^{-\alpha}$:
\begin{equation}
f_{\alpha}(k)\sim\frac{1}{k^{1-\alpha}},\,\alpha\in(0,\,1].\label{eq:f-alpha-ana-power}
\end{equation}
This result can be shown using the singularity of the polylogarithm
functions~\citep{olver2010nisthandbook,vodola2014kitaevchains}.
However, this approach does not allows to obtain general property
of $f_{\alpha}(k)$ when $\kappa_{l,\,\alpha}$ takes a more general
class of functions. Now we shall we shall explicitly show the singularity
of $f_{\alpha}(k)\sim1/k^{1-\alpha}$ around $k=0$ for $\alpha\in(0,\,1]$
without resorting to the polylogarithm functions. Recall 
\begin{equation}
f_{\alpha}(k)\equiv2\sum_{l=1}^{N/2-1}\kappa_{l,\,\alpha}\sin(kl)+\kappa_{N/2,\,\alpha}.\label{eq:f-alpha-def}
\end{equation}
Note that due to the regularity condition~(\ref{eq:kappa-reg-ends}),
we know that $\kappa_{N/2,\,\alpha}$ is finite as $N\to\infty$.
Therefore in what follows we shall omit $\kappa_{N/2,\,\alpha}$ in
the definition of $f_{\alpha}(k)$ because it does not affect the
analytic property of $f_{\alpha}(k)$. Now we are in a position to
prove Eq.~(\ref{eq:f-alpha-ana-power}):
\begin{proof}
Apparently $f_{0}(k)$ can be exactly calculated to be $\cot(k/2)$
which scales as $1/k$ near $k=0$. For the case $\alpha\in(0,\,1]$,
after applying the Euler-Maclaurin formula~(\ref{eq:E-M}), $f_{\alpha}(k)$
becomes 
\begin{equation}
f_{\alpha}(k)=2\mathscr{F}_{\alpha}(k)+\mathscr{R}_{\alpha}(k),\label{eq:falpha-EM}
\end{equation}
where 
\begin{equation}
\mathscr{F}_{\alpha}(k)=\int_{1}^{N/2-1}\frac{\sin(kx)}{x^{\alpha}}dx,
\end{equation}
and the remainder is 
\begin{equation}
\mathscr{R}_{\alpha}(k)=2k\int_{1}^{N}P_{1}(\{x\})\frac{\cos(kx)}{x^{\alpha}}-2\int_{1}^{N}P_{1}(\{x\})\frac{\sin(kx)}{x^{\alpha+1}},\label{eq:R-f}
\end{equation}
where we have again ignored the finite boundary terms. Apparently
the second term in Eq.~(\ref{eq:R-f}) is finite and therefore will
not contribute to the singularity of $f_{\alpha}(k)$, since 
\begin{equation}
2\bigg|\int_{1}^{N}P_{1}(\{x\})\frac{\sin(kx)}{x^{\alpha+1}}\bigg|<\bigg|\int_{1}^{N}\frac{1}{x^{\alpha+1}}\bigg|<\infty.
\end{equation}
Our goal now is to determine the asymptotic behavior of the first
term of Eq.~(\ref{eq:R-f}). Applying Fourier transform of $P_{1}(\{x\})$~\citep{knopp1990theoryand}
\begin{equation}
P_{1}(\{x\})=-\sum_{m=1}^{\infty}\frac{\sin(2m\pi x)}{m\pi},
\end{equation}
we obtain 
\begin{align}
\int_{1}^{N}P_{1}(\{x\})\frac{\cos(kx)}{x^{\alpha}}dx & =-\frac{1}{2}\sum_{m=1}^{\infty}\frac{1}{m\pi}\int_{1}^{N}\left\{ \frac{\sin[(2m\pi+k)x]}{x^{\alpha}}\right.\nonumber \\
 & +\left.\frac{\sin[(2m\pi-k)x]}{x^{\alpha}}\right\} dx.\label{eq:R-f-first}
\end{align}
Integrating by parts, we find that 
\begin{align}
\int_{1}^{N}\frac{\sin[(2m\pi+k)x]}{x^{\alpha}} & =\frac{1}{2m\pi+k}\left\{ \frac{\cos[(2m\pi+k)x]}{x^{\alpha}}\big|_{x=1}^{N}\right.\nonumber \\
 & +\left.\alpha\int_{1}^{N}\frac{\cos[(2m\pi+k)x]}{x^{\alpha+1}}dx\right\} .\label{eq:sin-power-scaling1}
\end{align}
Apparently, the integral on the r.h.s is bounded in the limit $N\to\infty$
as long as $\alpha>0$. So in the limit $N\to\infty$, we find 
\begin{equation}
\int_{1}^{N}\frac{\sin[(2m\pi+k)x]}{x^{\alpha}}\lesssim\frac{1}{2m\pi+k}.
\end{equation}
By similar argument, one can show 
\begin{equation}
\int_{1}^{N}\frac{\sin[(2m\pi-k)x]}{x^{\alpha}}\lesssim\frac{1}{2m\pi-k}.
\end{equation}
Note if $k$ is resonant with $2m\pi$ then the original integral
vanishes and there is no need to do the scaling analysis for the second
integral on the r.h.s. of Eq.~(\ref{eq:R-f-first}). Substituting
above results into Eq.~(\ref{eq:R-f-first}), we find 
\begin{equation}
\int_{1}^{N}P_{1}(\{x\})\frac{\cos(kx)}{x^{\alpha}}\lesssim\sum_{m=1}^{\infty}\frac{1}{m^{2}}<\infty.\label{eq:sin-power-scaling2}
\end{equation}

Up to now, we have shown that there is no singularity in the remainder
as long as $\alpha>0$. To see the singularity in the first term of
Eq.~(\ref{eq:falpha-EM}), we make change of variables $kl=s$ and
obtain 
\begin{equation}
\mathscr{F}_{\alpha}(k)=\frac{1}{k^{1-\alpha}}\int_{k}^{n\pi+\pi/2}\frac{\sin s}{s^{\alpha}}ds,\label{eq:integral-s}
\end{equation}
where we note $Nk=(2n+1)\pi$, where $n=0,\,1\cdots,\,N-1$. In the
limit $N\to\infty$, if finite $n$ is finite, apparently the singularity
of $f_{\alpha}(k)\sim1/k^{1-\alpha}$. On the other hand, if $n\to\infty$,
the integral in Eq.~(\ref{eq:integral-s}) is still finite since
\begin{equation}
\int_{0}^{\infty}ds\frac{\sin s}{s^{\alpha}}=\Gamma(1-\alpha)\cos\left(\frac{\pi\alpha}{2}\right),\text{for},\,\alpha\in(0,\,1]\label{eq:SIntegral-convergence}
\end{equation}
as we will now show. We note that

\begin{equation}
\lim_{n\to\infty}\int_{0}^{n\pi+\pi/2}\frac{\sin s}{s^{\alpha}}ds=\int_{0}^{\infty}ds\frac{\sin s}{s^{\alpha}}=\text{Im}\int_{0}^{\infty}dss^{-\alpha}e^{\text{i}s}.\label{eq:S-infty}
\end{equation}
One can evaluate Eq.~(\ref{eq:S-infty}) by first replacing $s\to\text{i}t$
and obtain 
\begin{align}
\int_{0}^{\infty}dss^{-\alpha}e^{\text{i}s} & =\text{i}^{1-\alpha}\int_{0}^{\text{i}\infty}dtt^{-\alpha}e^{-t}\nonumber \\
 & =\text{i}^{1-\alpha}\lim_{\varepsilon\to0}\left[\int_{0}^{\text{i}\varepsilon}dtt^{-\alpha}e^{-t}+\int_{\text{i}\varepsilon}^{\text{i}\infty}dtt^{-\alpha}e^{-t}\right].\label{eq:s-to-t}
\end{align}
The convergence of the first integral on the r.h.s. of Eq.~(\ref{eq:s-to-t})
requires that $\alpha<1$. The convergence of the second integral
on the r.h.s. of Eq.~(\ref{eq:s-to-t}) requires the integrand vanishes
at $t=\text{i}\infty$, which leads to $\alpha>0$. Now we take advantage
of the analyticity of the integrand for $\alpha\in(0,\,1]$ and rotate
the integral from positive imaginary $t-$axis to positive real $t-$axis,
which yields, 
\begin{equation}
\int_{0}^{\text{i}\infty}dtt^{-\alpha}e^{-t}=\int_{0}^{\infty}dtt^{-\alpha}e^{-t}=\Gamma(1-\alpha),
\end{equation}
which concludes the proof of Eq.~(\ref{eq:SIntegral-convergence})
for $\alpha\in(0,\,1)$. In fact Eq.~(\ref{eq:SIntegral-convergence})
also holds for $\alpha=1$ since $\int_{0}^{\infty}ds\sin s/s=\frac{\pi}{2}$
which can be evaluated by the residue theorem is actually $\lim_{\alpha\to1}\Gamma(1-\alpha)\cos\left(\frac{\pi\alpha}{2}\right)$.

Therefore, we have successfully shown that the singularity of $f_{\alpha}(k)$
only lies in the main term of the Euler-Maclaurin formula, which is
Eq.~(\ref{eq:falpha-EM}).
\end{proof}

\section{\label{sec:Proof-General-D}An integral approximation to $f_{\alpha}(k)$}

We show in Sec.~\ref{sec:Singularity-f-alpha-power} that the singularity
of $f_{\alpha}(k)$ when $\kappa_{l,\,\alpha}=l^{-\alpha}$ can be
explicitly found with only elementary techniques, without resorting
to the polylogarithmic function as in the original proposal of the
LRK~\citep{vodola2014kitaevchains}. The advantage of this approach
is that it will allows us to prove the following theorem for general
functions $\kappa_{l,\,\alpha}$ that that satisfy the regularity
conditions~(\ref{eq:kappa-reg-ends},~\ref{eq:kappa-regularity}):
\begin{thm}
\label{thm:General-d}We consider a general piecewise smooth function
$\kappa_{x,\,\alpha}$ that satisfies the regularity conditions in
the main text, i.e., $\kappa_{x,\,\alpha}$ satisfies (i) 
\begin{align}
\big|\kappa_{x,\,\alpha}^{(q)} & \big|<\infty,\,q=0,\,1,\,\cdots2Q,\label{eq:kappa-reg-ends}
\end{align}
which holds piecewisely on $[1,\,\infty]$, and

(ii) 
\begin{equation}
\big|\int_{1}^{\infty}\kappa_{x,\,\alpha}^{(2Q+1)}dx\big|<\infty,\label{eq:kappa-regularity}
\end{equation}
where $Q$ is a non-negative integer and the superscript $(q)$ denotes
the $q$-th derivative with respect to $x$.

Then the singularity of $f_{\alpha}(k)$ near $k=0$ is controlled
by the main integral in the Euler-Maclaurin formula, i.e., the first
term in 
\begin{equation}
f_{\alpha}(k)=2\mathscr{F}_{\alpha}(k)+\mathscr{R}_{\alpha}(k).\label{eq:f-rep}
\end{equation}
\end{thm}

Before we start the proof, let us first note that for the long-range
decay function $d_{x,\,\alpha}$, we allow not only smooth functions
of $x$, but also piecewise functions consisting of several smooth
functions. This is because, as we have seen in Sec.~\ref{sec:Euler-Maclaurin-formula}
and \ref{sec:Proof-Scaling-Gamma}, one can apply the Euler-Maclaurin
in a piecewise way. The condition~(\ref{eq:kappa-reg-ends}) indicates
there can be only discontinuities at the joints, but no singularities.
Nevertheless, in what follows, we shall prove for the case when $d_{x,\,\alpha}$
is smooth $[1,\,\infty]$, which can be easily generalized to the
case of piecewise smoothness without any difficulty.
\begin{proof}
We take $M=Q$ in the Euler-Maclaurin formula~(\ref{eq:E-M}), and
obtain 
\begin{equation}
f_{\alpha}(k)=2\mathscr{F}_{\alpha}(k)+\mathscr{R}_{\alpha}(k),
\end{equation}
where 
\begin{equation}
\mathscr{F}_{\alpha}(k)=\int_{1}^{N/2-1}\sin(kx)\kappa_{x,\,\alpha}dx,
\end{equation}
\begin{equation}
\mathscr{R}_{\alpha}(k)=\sum_{q=0}^{2Q+1}\mathscr{R}_{\alpha,\,q}(k)+\sum_{m=1}^{Q}\frac{b_{2m}}{(2m)!}\left[\sin(kx)\kappa_{x,\,\alpha}\right]^{(2m-1)}\bigg|_{x=1}^{x=N/2-1},
\end{equation}
\begin{equation}
\mathscr{R}_{\alpha,\,q}(k)\equiv\mathcal{C}_{Q,\,q}\int_{1}^{N/2-1}P_{2Q+1}(x)[\sin(kx)]^{(q)}\left[\kappa_{x,\,\alpha}\right]^{(2Q+1-q)}dx,\label{eq:scrR-q}
\end{equation}
\begin{equation}
\mathcal{C}_{Q,\,q}\equiv\begin{pmatrix}2Q-1\\
q
\end{pmatrix}\frac{1}{(2Q+1)!}.
\end{equation}
Apparently, the boundary term is finite due to the regularity condition~(\ref{eq:kappa-reg-ends}).
Thus we shall focus on the integral in the remainder $\mathscr{R}_{\alpha}(k)$
subsequently. For $q=0$, i we find 
\begin{align}
\mathscr{R}_{\alpha,\,q}(k)\leq & \mathcal{C}_{Q,\,q}\mathcal{C}_{2Q+1}\bigg|\int_{1}^{N/2-1}\kappa_{x,\,\alpha}^{(2Q+1)}dx\bigg|<\infty,
\end{align}
where we have used the fact that $P_{2Q+1}(x)$ is bounded, $\big|P_{2Q+1}(x)\big|\le\mathcal{C}_{2Q+1}$
and Eq.~(\ref{eq:kappa-regularity}). When $q\ge1$, we apply the
Fourier transform of $P_{2Q+1}(x)$~\citep{knopp1990theoryand} 
\begin{equation}
P_{2Q+1}(x)=\sum_{m=1}^{\infty}(-1)^{Q-1}\frac{2\sin(2m\pi x)}{(2m\pi)^{2Q+1}}.
\end{equation}
to Eq.~(\ref{eq:scrR-q}). For for even $q>0$, we obtain 
\begin{gather}
\mathscr{R}_{\alpha,\,q}(k)=\mathcal{C}_{Q,\,q}k^{q}(-1)^{Q-1+q/2}\sum_{m=1}^{\infty}\frac{1}{(2m\pi)^{2Q+1}}\times\nonumber \\
\int_{1}^{N/2-1}\left(\cos[(2m\pi-k)x]-\cos[(2m\pi+k)x]\right)\kappa_{x,\,\alpha}^{(2Q+1-q)}dx,\label{eq:scrR-even}
\end{gather}
and for odd $q\ge1$, we obtain 
\begin{gather}
\mathscr{R}_{\alpha,\,q}(k)=\mathcal{C}_{Q,\,q}k^{q}(-1)^{Q-1+(q-1)/2}\sum_{m=1}^{\infty}\frac{1}{(2m\pi)^{2Q+1}}\times\nonumber \\
\int_{1}^{N/2-1}\left\{ \sin[(2m\pi+k)x]+\sin[(2m\pi-k)x]\right\} \kappa_{x,\,\alpha}^{(2Q+1-q)}dx.\label{eq:scrR-odd}
\end{gather}
Since the convergence of the series is determined by the behavior
of the general term at large values of the index, we shall focus on
the case of large $m$ in the series in Eqs.~(\ref{eq:scrR-even},
\ref{eq:scrR-odd}) subsequently. Similarly with Eqs.~(\ref{eq:sin-power-scaling1}-\ref{eq:sin-power-scaling2}),
one can perform integration by parts until one gets an integrand that
contains $(\kappa_{x,\,\alpha})^{(2Q+1)}$, which yields 
\begin{align}
 & \int_{1}^{N/2-1}\cos[(2m\pi\pm k)x]\left(\kappa_{x,\,\alpha}\right)^{(2Q+1-q)}dx\nonumber \\
= & \frac{1}{(2m\pi\pm k)}\left\{ \sin[(2m\pi\pm k)x]\left(\kappa_{x,\,\alpha}\right)^{(2Q+1-q)}\big|_{x=1}^{x=\infty}\right.\nonumber \\
+ & \left.\frac{1}{(2m\pi\pm k)}\left[\cos[(2m\pi\pm k)x]\left(\kappa_{x,\,\alpha}\right)^{(2Q+2-q)}\big|_{x=1}^{x=\infty}\right]+\cdots\right\} \nonumber \\
+ & \frac{(-1)^{(q-1)/2}}{(2m\pi\pm k)^{q}}\int_{1}^{N/2-1}\cos[(2m\pi\pm k)x]\kappa_{x,\,\alpha}^{(2Q+1)}dx,
\end{align}
Apparently the last integral is bounded, according to Eq.~(\ref{eq:kappa-regularity}).
Therefore, for large $m$ we know 
\begin{align}
\bigg|\int_{1}^{N/2-1}\cos[(2m\pi\pm k)x]\left(\kappa_{x,\,\alpha}\right)^{(2Q+1-q)}dx\bigg| & \lesssim\frac{\mathcal{C}_{2Q+1-q}(\alpha)}{(2m\pi\pm k)},
\end{align}
where 
\begin{equation}
\mathcal{C}_{q}(\alpha)\equiv\big|\left(\kappa_{x,\,\alpha}^{(q)}\right)_{x=1}\big|+\big|\left(\kappa_{x,\,\alpha}^{(q)}\right)_{x=\infty}\big|.
\end{equation}
which is finite according to the regularity condition Eq.~(\ref{eq:kappa-reg-ends}).
With similar argument, we obtain 
\begin{equation}
\bigg|\int_{1}^{N/2-1}\sin[(2m\pi\pm k)x]\kappa_{x,\,\alpha}^{(2Q+1-q)}dx\bigg|\lesssim\frac{\mathcal{C}_{2Q+1-q}(\alpha)}{(2m\pi\pm k)}
\end{equation}
for large $m$. So we conclude 
\begin{equation}
\big|\mathscr{R}_{\alpha,\,q}(k)\big|\lesssim\mathcal{C}_{Q,\,q}\mathcal{C}_{2Q+1-q}(\alpha)k^{q}\sum_{m=1}^{\infty}\frac{1}{(2m\pi)^{2Q+2}},
\end{equation}
which is bounded for all finite values of $k$. Now we have shown
that the remainder $R_{\alpha,\,q}(k)$ is regular with no singularity
in $k$, which completes the proof.
\end{proof}

\section{\label{sec:Singularity-f}The possible singularity of $\mathscr{F}_{\alpha}(k)$
or $f_{\alpha}(k)$ for $\kappa_{l,\,\alpha}$ satisfying the regularity
conditions}
\begin{cor}
\label{cor:sing-at-zero}

If $\kappa_{l,\,\alpha}$ satisfies the regularity conditions ~(\ref{eq:kappa-reg-ends}-\ref{eq:kappa-regularity}),
then $\mathscr{F}_{\alpha}(k)$ is regular as long as $k\neq0$.
\end{cor}

\begin{proof}
The proof is straightforward: performing integration by part for the
main integral in Eq.~(\ref{eq:f-rep}), we find 
\begin{gather}
\int_{1}^{N/2-1}\sin(kx)\kappa_{x,\,\alpha}dx=-\frac{\cos(kx)}{k}\kappa_{x,\,\alpha}|_{x=1}^{x=N/2-1}\nonumber \\
+\frac{\sin(kx)}{k^{2}}\kappa_{x,\,\alpha}^{(1)}\big|_{x=1}^{x=N/2-1}\cdots+\frac{(-1)^{Q}}{k^{2Q}}\int_{1}^{N/2-1}\cos(kx)\kappa_{x,\,\alpha}^{(2Q+1)}.
\end{gather}
Since the integral is bounded according to Eq.~(\ref{eq:kappa-regularity}),
we find 
\begin{equation}
\big|\mathscr{F}_{\alpha}(k)\big|\leq\sum_{m=1}^{2Q}\frac{\mathcal{C}_{m-1}(\alpha)}{k^{m}}.
\end{equation}
Thus, $f_{\alpha}(k)$ or $\mathscr{F}_{\alpha}(k)$ are bounded as
long as $k\neq0$.
\end{proof}
From this proof, we immediately see that the only possible singularity
of $\mathscr{F}_{\alpha}(k)$ is at $k=0$. As we have mentioned in
the main text, we can introduce a trick to get a rough estimate about
the possible singularity of the main integral $\mathscr{F}_{\alpha}(k)$
near $k=0$. We integrate over $k$ from $1/N$ to $\Lambda$, where
$\Lambda$ is finite. This yields 
\begin{equation}
\int_{1/N}^{\Lambda}dk\mathscr{F}_{\alpha}(k)=\int_{1}^{N/2-1}\frac{\kappa_{x,\,\alpha}}{x}dx-\int_{1}^{N/2-1}\frac{\kappa_{x,\,\alpha}\cos(\Lambda x)}{x}dx,\label{eq:int-scrF}
\end{equation}
where we have interchanged the order of integration. According to
Sec.~\ref{sec:The-convergence-integration-by-parts}, the second
integral on the r.h.s. of Eq.~(\ref{eq:int-scrF}) is bounded and
the exact scaling with respect to $N$ can be easily found by integrating
by parts. Therefore the scaling of $\int_{1/N}^{\Lambda}\mathscr{F}_{\alpha}(k)dk$
is totally controlled by the first integral on the r.h.s of Eq.~(\ref{eq:int-scrF}).
If the scaling of $\int_{1/N}^{\Lambda}\mathscr{F}_{\alpha}(k)dk$
can be computed, it can reveal \textit{some partial} information about
the singularity of $\mathscr{F}_{\alpha}(k)$ around $k=0$. For example,
if $\int_{1/N}^{\Lambda}\mathscr{F}_{\alpha}(k)dk\sim\text{constant}$
or $\int_{1/N}^{\Lambda}\mathscr{F}_{\alpha}(k)dk\sim\ln N$, then
we know the singularity of $\mathscr{F}_{\alpha}(k)$ at $k=0$ is
at most $1/k^{1-\varepsilon}$ or $1/k$ respectively, where $\varepsilon$
is arbitrary small positive number.

\section{\label{sec:The-convergence-integration-by-parts}The convergence
of the integral $\int_{1}^{\infty}\kappa_{x,\,\alpha}\cos(\Lambda x)/x]dx$}

One can prove the integral $\int\kappa_{x,\,\alpha}\cos(\Lambda x)/x]dx$
is bounded via integration by parts. First, it is found checked that
\begin{gather}
\int_{1}^{\infty}\kappa_{x,\,\alpha}\cos(\Lambda x)/x]dx=\frac{1}{\Lambda}\sin(\Lambda x)\frac{\kappa_{x,\,\alpha}}{x}\big|_{x=1}^{x=\infty}\nonumber \\
+\frac{1}{\Lambda}\int_{1}^{\infty}\sin(\Lambda x)\frac{\kappa_{x,\,\alpha}}{x^{2}}-\frac{1}{\Lambda}\int_{1}^{\infty}\sin(\Lambda x)\frac{\kappa_{x,\,\alpha}^{(1)}}{x}.
\end{gather}
According to regularity condition~(\ref{eq:kappa-reg-ends}) of $\kappa_{x,\,\alpha}$,
we know 
\begin{equation}
\bigg|\int_{1}^{\infty}\sin(\Lambda x)\frac{\kappa_{x,\,\alpha}}{x^{2}}\bigg|\le\mathcal{C}_{\max,\,\alpha}\bigg|\int_{1}^{\infty}\frac{1}{x^{2}}\bigg|<\infty,\label{eq:q-zero}
\end{equation}
where 
\begin{equation}
\mathcal{C}_{\max,\,\alpha}\equiv\max_{x\in[1,\,\infty]}\kappa_{x,\,\alpha}.
\end{equation}
The convergence of $\int_{1}^{\infty}\kappa_{x,\,\alpha}\cos(\Lambda x)/x]dx$
will depends on the convergence $\int_{1}^{\infty}\sin(\Lambda x)\kappa_{x,\,\alpha}^{(1)}/x$.
Further integrating by parts and applying the same argument, one can
show that the convergence of $\int_{1}^{\infty}\sin(\Lambda x)\kappa_{x,\,\alpha}^{(1)}/x$,
will depend on $\int_{1}^{\infty}\cos(\Lambda x)\kappa_{x,\,\alpha}^{(2)}/x$.
We continue to apply integration by parts until we obtain $\int_{1}^{\infty}\sin(\Lambda x)\kappa_{x,\,\alpha}^{(2Q+1)}/x$,
which yields,

\begin{align}
 & \int_{1}^{\infty}\kappa_{x,\,\alpha}\cos(\Lambda x)/x]dx\nonumber \\
= & \frac{1}{\Lambda}\sin(\Lambda x)\frac{\kappa_{x,\,\alpha}}{x}\big|_{x=1}^{x=\infty}+\frac{1}{\Lambda}\int_{1}^{\infty}\sin(\Lambda x)\frac{\kappa_{x,\,\alpha}}{x^{2}}\\
+ & \frac{1}{\Lambda^{2}}\cos(\Lambda x)\frac{\kappa_{x,\,\alpha}^{(1)}}{x}\big|_{x=1}^{x=\infty}+\frac{1}{\Lambda^{2}}\int_{1}^{\infty}\cos(\Lambda x)\frac{\kappa_{x,\,\alpha}^{(1)}}{x^{2}}\nonumber \\
+ & \cdots+\frac{1}{\Lambda^{2Q+2}}\int_{1}^{\infty}\cos(\Lambda x)\frac{\kappa_{x,\,\alpha}^{(2Q+1)}}{x^{2}}.
\end{align}
Once again all the boundary terms in the above equation are bounded
thanks to the regularity condition~(\ref{eq:kappa-reg-ends}). Furthermore,
as with Eq.~(\ref{eq:q-zero}), we note, 
\begin{equation}
\bigg|\int_{1}^{\infty}\sin(\Lambda x)\frac{\kappa_{x,\,\alpha}^{(2q)}}{x^{2}}\bigg|\le\mathcal{C}_{\max,\,\alpha}^{(2q)}\bigg|\int_{1}^{\infty}\frac{1}{x^{2}}\bigg|<\infty,
\end{equation}
\begin{equation}
\bigg|\int_{1}^{\infty}\cos(\Lambda x)\frac{\kappa_{x,\,\alpha}^{(2q+1)}}{x^{2}}\bigg|\le\mathcal{C}_{\max,\,\alpha}^{(2q+1)}\bigg|\int_{1}^{\infty}\frac{1}{x^{2}}\bigg|<\infty,
\end{equation}
where $q=1,2,\,\cdots,\,Q$ and 
\begin{equation}
\mathcal{C}_{\max,\,\alpha}^{(q)}\equiv\max_{x\in[1,\,\infty]}\big|\kappa_{x,\,\alpha}^{(q)}\big|.
\end{equation}
Thus the convergence of $\int_{1}^{\infty}\kappa_{x,\,\alpha}\cos(\Lambda x)/x]dx$
depends on $\int_{1}^{\infty}\sin(\Lambda x)\kappa_{x,\,\alpha}^{(2Q+1)}/x$.
We use 
\begin{equation}
\bigg|\int_{1}^{\infty}\cos(\Lambda x)\frac{\kappa_{x,\,\alpha}^{(2Q+1)}}{x^{2}}\bigg|\le\bigg|\int_{1}^{\infty}\kappa_{x,\,\alpha}^{(2Q+1)}\bigg|<\infty
\end{equation}
according to the regularity condition~(\ref{eq:kappa-regularity}).
We conclude that the integral $\int_{1}^{\infty}\kappa_{x,\,\alpha}\cos(\Lambda x)/x]dx$
is convergent.

According to Theorem~\ref{thm:General-d}, we find $f_{\alpha}(k)\sim\mathscr{F}_{\alpha}(k)\equiv\int_{1}^{N/2-1}\sin(kx)\kappa_{x,\,\alpha}dx$.
Furthermore, according to Corollary~\ref{cor:sing-at-zero}, the
only possible singularity of $\mathscr{F}_{\alpha}(k)$ is near $k=0$.
Using trick in Eq.~(\ref{eq:int-scrF}), one can obtain some information
about the behavior of $\mathscr{F}_{\alpha}(k)$ around $k=0$ by
investigating the scaling of $\int_{1/N}^{\Lambda}\mathscr{F}_{\alpha}(k)dk$
with $\Lambda$ being any finite number. As we have shown above, the
second term on the r.h.s. of Eq.~(\ref{eq:int-scrF}) is convergent,
the scaling of $\int_{1/N}^{\Lambda}\mathscr{F}_{\alpha}(k)dk$ with
respect to $N$ is the same as the one of $\int_{1}^{N}\kappa_{x,\,\alpha}/xdx$.
An immediate consequence is that the singularity of $\mathscr{F}_{\alpha}(k)$
is at most $\mathcal{O}\left(1/k\right)$ since $\int_{1}^{N}(\kappa_{x,\,\alpha}/x)dx\le\mathcal{C}_{\max,\,\alpha}\int_{1}^{N}1/xdx\sim\ln N$,
where $\mathcal{C}_{\max,\,\alpha}=\max_{x\in[1,\,\infty]}\kappa_{x,\,\alpha}$.

\section{\label{sec:Finite-size-scaling}Finite-size scaling}

We note that Eq.~(\ref{eq:I0-Delta-scaling}) in the main text gives
the asymptotic scaling of $I_{0}(\Delta)$ in the thermodynamics limit
$N\to\infty$. For $\kappa_{x,\,\alpha}=x^{-\alpha}$, super-HS transition
only occurs at $\alpha=0$ for $N\to\infty$. However, for large but
finite $N$, small $\alpha$ near zero may also lead the super-HS,
which we now discuss. Setting $\alpha=\epsilon$, where $\epsilon$
is a small number, we obtain 
\begin{align}
\int_{1}^{N}\frac{dx}{x^{1+\epsilon}} & =\int_{1}^{N}dx\frac{1}{x}e^{-\ln x\epsilon}\nonumber \\
 & =\int_{1}^{N}dx\frac{1}{x}\left[\sum_{n=0}^{\infty}\frac{(-1)^{n}\epsilon^{n}(\ln x)^{n}}{n!}\right]\nonumber \\
 & =\sum_{n=0}^{\infty}\frac{(-1)^{n}\epsilon^{n}}{n!}\int_{1}^{N}dx\frac{(\ln x)^{n}}{x}=\mathcal{S}(\epsilon\ln N)\ln N,
\end{align}
where 
\begin{equation}
\mathcal{S}(a)\equiv\sum_{n=0}^{\infty}\frac{(-1)^{n}a^{n}}{(n+1)!}.
\end{equation}
Therefore we find when 
\begin{equation}
\epsilon\ln N\ll1
\end{equation}
$\mathcal{S}(\epsilon\ln N)\to1$, so that 
\begin{equation}
\int_{1}^{N}\frac{dx}{x^{1+\epsilon}}\sim\ln N.\label{eq:finiteN}
\end{equation}
Alternatively, $\mathcal{S}(a)$ may be evaluated exactly, which is
\begin{equation}
\mathcal{S}(a)=-\frac{1}{a}\sum_{n=1}^{\infty}\frac{(-1)^{n}a^{n}}{n!}=\frac{1}{a}(1-e^{-a}).
\end{equation}
From which one can clearly see that $\mathcal{S}(\epsilon\ln N)\to1$
as $\epsilon\ln N\to0$. Therefore, according to Eq.~(\ref{eq:I0-Delta-scaling})
in the main text, we see that for $\kappa_{x,\,\epsilon}=x^{-\epsilon}$,
we have 
\begin{equation}
I_{0}(\Delta)\sim N^{2}(\ln N)^{2},\,\text{for}\,\epsilon\ll(\ln N)^{-1}.
\end{equation}
By similar analysis, one can show analogously that for $\kappa_{x,\,1+\epsilon}=(1+\ln x)^{-(1+\epsilon)}$
\begin{equation}
I_{0}(\Delta)\sim N^{2}(\ln\ln N)^{2},\,\text{for}\,\epsilon\ll(\ln\ln N)^{-1}.
\end{equation}

\section{\label{sec:spin-rep}The LRK Hamiltonian in the spin representation}

With the Jordan-Wigner transformation~\citep{coleman2015introduction},

\begin{equation}
a_{j}^{\dagger}=(-1)^{j-1}\prod_{k=1}^{j-1}\sigma_{k}^{z}\sigma_{j}^{+},\label{eq:JW1}
\end{equation}
\begin{equation}
a_{j}=(-1)^{j-1}\prod_{k=1}^{j-1}\sigma_{k}^{z}\sigma_{j}^{-},\label{eq:JW2}
\end{equation}
where $\sigma_{j}^{z}$ is the standard Pauli $z$-matrix 
\begin{align}
\sigma_{j}^{+} & \equiv\begin{bmatrix}0 & 1\\
0 & 0
\end{bmatrix},\\
\sigma_{j}^{-} & \equiv\begin{bmatrix}0 & 0\\
1 & 0
\end{bmatrix},
\end{align}
it is readily checked that 
\begin{align}
a_{j}^{\dagger}a_{j} & =\sigma_{j}^{+}\sigma_{j}^{-}=\frac{1}{2}(\sigma_{j}^{z}+1),
\end{align}
\begin{align}
a_{j+1}^{\dagger}a_{j} & =-\sigma_{j+1}^{+}\sigma_{j}^{z}\sigma_{j}^{-}=\sigma_{j+1}^{+}\sigma_{j}^{-},
\end{align}
where we have used the fact $\sigma_{j}^{z}\sigma_{j}^{\pm}=\pm\sigma_{j}^{\pm}$
in the second equation. Furthermore, 
\begin{align}
a_{j}a_{j+l} & =(-1)^{j-1}\prod_{k=1}^{j-1}\sigma_{k}^{z}\sigma_{j}^{-}\times(-1)^{j+l-1}\prod_{m=1}^{j-1+l}\sigma_{m}^{z}\sigma_{j+l}^{-}\nonumber \\
 & =(-1)^{l}\sigma_{j}^{-}\sigma_{j}^{z}\prod_{k=j+1}^{j-1+l}\sigma_{k}^{z}\sigma_{j+l}^{-}\nonumber \\
 & =(-1)^{l}\sigma_{j}^{-}\prod_{k=j+1}^{j-1+l}\sigma_{k}^{z}\sigma_{j+l}^{-},
\end{align}
where we have used the fact that $\sigma_{j}^{\pm}\sigma_{j}^{z}=\mp\sigma_{j}^{\pm}$.
Now using the relation 
\begin{align}
\sigma_{j}^{+} & \equiv\frac{1}{2}(\sigma_{j}^{x}+\text{i}\sigma_{j}^{y}),\\
\sigma_{j}^{-} & \equiv\frac{1}{2}(\sigma_{j}^{x}-\text{i}\sigma_{j}^{y}),
\end{align}
where $\sigma_{j}^{x}$ and $\sigma_{j}^{y}$ are standard Pauli $x-$
and $y-$ matrices respectively, we find \begin{widetext}
\begin{align}
\sigma_{j}^{+}\sigma_{j+l}^{-}+\sigma_{j}^{-}\sigma_{j+l}^{+} & =\frac{1}{4}(\sigma_{j}^{x}+\text{i}\sigma_{j}^{y})(\sigma_{j+l}^{x}-\text{i}\sigma_{j+l}^{y})+\frac{1}{4}(\sigma_{j}^{x}-\text{i}\sigma_{j}^{y})(\sigma_{j+l}^{x}+\text{i}\sigma_{j+l}^{y})\nonumber \\
 & =\frac{1}{2}(\sigma_{j}^{x}\sigma_{j+l}^{x}+\sigma_{j}^{y}\sigma_{j+l}^{y}),\label{eq:Pauli-Id1}
\end{align}
\begin{align}
\sigma_{j}^{+}\sigma_{j+l}^{+}+\sigma_{j}^{-}\sigma_{j+l}^{-} & =\frac{1}{4}(\sigma_{j}^{x}+\text{i}\sigma_{j}^{y})(\sigma_{j+l}^{x}+\text{i}\sigma_{j+l}^{y})+\frac{1}{4}(\sigma_{j}^{x}-\text{i}\sigma_{j}^{y})(\sigma_{j+l}^{x}-\text{i}\sigma_{j+l}^{y})\nonumber \\
 & =\frac{\text{1}}{2}(\sigma_{j}^{x}\sigma_{j+l}^{x}-\sigma_{j}^{y}\sigma_{j+l}^{y}).\label{eq:Pauli-Id2}
\end{align}
Using Eqs.~(\ref{eq:JW1},~\ref{eq:JW2},~\ref{eq:Pauli-Id1}),
the tunneling and kinetic terms become 
\begin{equation}
\sum_{j=1}^{N}(a_{j}^{\dagger}a_{j+1}+a_{j+1}^{\dagger}a_{j})=\sum_{j=1}^{N}(\sigma_{j}^{-}\sigma_{j+1}^{+}+\sigma_{j}^{+}\sigma_{j+1}^{-})=\frac{1}{2}\sum_{j=1}^{N}(\sigma_{j}^{x}\sigma_{j+l}^{x}+\sigma_{j}^{y}\sigma_{j+l}^{y}),\label{eq:tunnel-spin}
\end{equation}
and
\begin{equation}
\sum_{j=1}^{N}(a_{j}^{\dagger}a_{j}-\frac{1}{2})=\frac{1}{2}\sum_{j=1}^{N}\sigma_{j}^{z}\label{eq:kinetic-spin},
\end{equation}
respectively. For the long-range superconducting terms, with the anti-periodic
boundary condition, one can easily obtain the following alternative
form 
\begin{equation}
\sum_{j=1}^{N-1}\sum_{l=1}^{N-j}\kappa_{l,\,\alpha}a_{j}a_{j+l}=\frac{1}{2}\sum_{j=1}^{N}\sum_{l=1}^{N-1}\kappa_{l,\,\alpha}a_{j}a_{j+l},
\end{equation}
and and a similar equation for the term $\sum_{j=1}^{N-1}\sum_{l=1}^{N-j}\kappa_{l,\,\alpha}a_{j+l}^{\dagger}a_{j}^{\dagger}$.
On the other hand, with Eqs.~(\ref{eq:JW1},~\ref{eq:JW2},~\ref{eq:Pauli-Id2}),
we find

\begin{equation}
\sum_{j=1}^{N}\sum_{l=1}^{N-1}\kappa_{l,\,\alpha}(a_{j}a_{j+l}+a_{j+l}^{\dagger}a_{j}^{\dagger})=\frac{1}{2}\sum_{j=1}^{N}\sum_{l=1}^{N-1}(-1)^{l}\kappa_{l,\,\alpha}(\sigma_{j}^{x}\sigma_{j+l}^{x}-\sigma_{j}^{y}\sigma_{j+l}^{y})\sigma_{j+1}^{z}\cdots\sigma_{j+l-1}^{z}.\label{eq:superconducting-spin}
\end{equation}
Substituting Eqs.~(\ref{eq:tunnel-spin}-\ref{eq:superconducting-spin})
into Eq.~(\ref{eq:LRK}) in the main text yields the LRK Hamiltonian
in the spin-representation, i.e., Eq.~\eqref{eq:Hspin} in the main
text. \end{widetext}

\bibliographystyle{apsrev4-1}
\bibliography{SuperHS-Kitaev}

\end{document}